\documentclass[11pt,orivec,letterpaper]{article}

\usepackage[left=1in,right=1in,top=1in,bottom=1in]{geometry}

\usepackage{bold-extra} 

\usepackage{tensor} 
\usepackage{mdwlist}
\usepackage{amsthm,amsfonts}
\usepackage[usenames]{color}
\usepackage{times}

\newtheorem{theorem}{Theorem}
\newtheorem{lemma}[theorem]{Lemma}
\newtheorem{claim}[theorem]{Claim}

\newtheorem{conjecture}[theorem]{Conjecture}
\newtheorem{proposition}[theorem]{Proposition}
\newtheorem{corollary}[theorem]{Corollary}
\newtheorem{definition}{Definition}
\newtheorem{remark}{Remark}

\newcommand{\comment}[1]{}

\usepackage{comment}
\usepackage{epsfig}
\usepackage{subfigure}
\usepackage{amsmath,amssymb,url}

\usepackage{paralist}

\usepackage{booktabs} 
\usepackage{multirow} 

\definecolor{Darkblue}{rgb}{0,0,0.4}
\definecolor{Brown}{cmyk}{0,0.81,1.,0.60}
\definecolor{Purple}{cmyk}{0.45,0.86,0,0}
\newcommand{\lref}[2][]{\hyperref[#2]{#1~\ref*{#2}}}

\renewenvironment{proof}{\begin{trivlist} \item[\hspace{\labelsep}{\bf
\noindent Proof.\/}] }{\qed\end{trivlist}}%

\usepackage{abstract}

\newcounter{note}

\begin{document}

\newcommand{\eweight}[1]{\text{\scshape Exact-Weight-$#1$}}
\newcommand{\ew}[1]{\text{\scshape Exact-$#1$}}
\def\sat{\text{\scshape 3-SAT}}
\def\ews{\text{exact-weight}}
\def\ksumn{\text{\scshape $k$-sum$^n$}}
\def\ksumN{\text{\scshape $k$-sum$^N$}}
\newcommand{\hpg}[1]{\text{$#1$-partite graph}}
\newcommand{\ksum}[1]{\text{\scshape $#1$-sum}}
\newcommand{\kxor}[1]{\text{\scshape $#1$-xor}}
\newcommand{\tableksum}[1]{\text{\scshape \text{$#1$}-sum}}
\newcommand{\ksol}[1]{\text{\text{\text{$#1$}-solution}}}
\newcommand{\convksum}[1]{\text{\scshape convolution-\text{$#1$}-sum}}
\def\sn{\text{super-node}}
\def\se{\text{super-edge}}
\newcommand{\ksumconj}[1]{\text{$\ksum{#1}$ Conjecture}}

\def\sepalg{\text{separator algorithm}}

\def\kpsum{\text{\scshape ($k$\textrm{+}$1$)-sum}}
\def\kmsum{\text{\scshape ($k$\textrm{-}$1$)-sum}}
\def\convkpsum{\text{\scshape convolution-($k$\textrm{+}$1$)-sum}}

\newcommand{\kmatching}[1]{{#1 { \text{\scshape\text{-matching}}}}}
\newcommand{\kstar}[1]{{#1 { \text{\scshape\text{-star}}}}}
\newcommand{\kpath}[1]{{#1 { \text{\scshape\text{-path}}}}}
\newcommand{\kcycle}[1]{{#1 { \text{\scshape\text{-cycle}}}}}
\newcommand{\kclique}[1]{{#1 { \text{\scshape -clique}}}}
\def\triangle{ { {\text{\scshape $3$-clique}}}}
\def\VI{\text{\scshape VI}}

\newcommand{\kmpath}{{{ \text{\scshape\text{($k$\textrm{-}$1$)-path}}}}}

\def \vertminor{\text{vertex-minor}} 
\def\vm{\leq_{\mathsf{vm}}}

\def\threematching{\kmatching{3}}
\def\fivepath{\kpath{5}}

\def\patrascu{\text{P{\v a}tra{\c s}cu}}

\def\P{P}

\newcommand{\ewtri}{\text{\scshape Exact-Weight-Triangle}}
\newcommand{\apsp}{\text{\scshape All-Pairs Shortest Paths}}

\def\legal{\text{Legal}}
\def\exactw{\text{Exact-Weight}}
\def\inject{\text{Injective}}
\def\wellf{\text{Well-Formed}}

\setcounter{tocdepth}{4}

\def\index{\mathsf{index}}
\def\poly{\mathsf{poly}}
\def\yes{\textsf{YES}}
\def\no{\textsf{NO}}

\newcommand{\hsubv}[2]{ { v_{{#1}, {#2}_{#1}} } } 
\newcommand{\hsub}[3]{ \{ \hsubv{#1}{#2} \}_{{#1} \in [{#3}]} } 
\newcommand{\hsubs}[3]{ \{ \hsubv{#1}{#2}
	\}_{h_{#1} \in {#3}} } 
\newcommand{\hsubsia}[1]{ \hsubs{i}{a}{#1} } 

\newcommand{\ksols}[3]{ \{ x_{{#1}, {#2}_{#1}} \}_{{#1} \in [{#3}]} }
\newcommand{\ksolsia}[1]{ \ksols{i}{a}{#1} }

\newcommand{\ksolsuper}[4]{ \{ x^{(#1)}_{{#2}, {#3}_{#2}} \}_{{#2} \in [{#4}]} }

\newcommand{\rto}[2]{\ensuremath{\tensor[_{#1}]{\leq}{_{#2}}}}
\newcommand{\rtont}[2]{\ensuremath{\tensor[_{#1}]{\lneq}{_{#2}}}}
\newcommand{\equivto}[2]{\ensuremath{\tensor[_{#1}]{\equiv}{_{#2}}}}
\newcommand{\cftwo}[1]{\left\lceil \frac{#1}{2} \right\rceil}
\newcommand{\ctwo}[1]{\left\lceil #1/2 \right\rceil}
\newcommand{\fftwo}[1]{\left\lfloor \frac{#1}{2} \right\rfloor}
\newcommand{\ftwo}[1]{\left\lfloor #1/2 \right\rfloor}

\def\kone{k^{(1)}}
\def\ktwo{k^{(2)}}
\def\khat{{\hat{k}}}

\def\polylog{\textsf{polylog}}

\def\ith{i^{th}}
\def\jth{j^{th}}
\def\kth{k^{th}}

\def\argmin{\textsf{argmin}}
\def\sepset{\mathcal{S}}

\def\Ot{\tilde{O}}
\def\ot{\tilde{o}}
\def\Omegat{\tilde{\Omega}}

\renewcommand{\arraystretch}{1.2}
\def\vspacer{\phantom{}}
\def\mrany{\multirow{2}{*}}

\renewcommand{\paragraph}[1]{ \vspace{0.175in} {\noindent \bf #1} }

\title{Exact Weight Subgraphs and the $k$-Sum Conjecture}

\urldef{\mailaa}\path|abboud@cs.stanford.edu|
\urldef{\mailkl}\path|klewi@cs.stanford.edu|


\author{Amir Abboud\thanks{Computer Science Department, Stanford University. 
\mailaa. This work was done while the author was supported by NSF grant CCF-1212372.} \and Kevin Lewi\thanks{Computer Science Department, Stanford 
University. \mailkl}}

\date{}

\maketitle

\setlength{\absleftindent}{8mm}
\setlength{\absrightindent}{8mm}

\begin{abstract}
	{ \small

	We consider the $\eweight{H}$ problem of finding a (not necessarily induced) 
	subgraph $H$ of weight $0$ in an edge-weighted graph $G$. We show that for 
	every $H$, the complexity of this problem is strongly related to that of the 
	infamous $\ksum{k}$ problem. In particular, we show that under the $\ksum{k}$ 
	Conjecture, we can achieve tight upper and lower bounds for the $\eweight{H}$ 
	problem for various subgraphs $H$ such as matching, star, path, and cycle.
	
	One interesting consequence is that improving on the $O(n^3)$ upper bound for 
	$\eweight{\kpath{4}}$ or $\eweight{\kpath{5}}$ will imply improved algorithms 
	for $\ksum{3}$, $\ksum{5}$, $\apsp$ and other fundamental problems. This is in 
	sharp contrast to the minimum-weight and (unweighted) detection versions, 
	which can be solved easily in time $O(n^2)$.
	We also show that a faster algorithm for any of the following three problems 
	would yield faster algorithms for the others: $\ksum{3}$, 
	$\eweight{\kmatching{3}}$, and $\eweight{\kstar{3}}$.
	}
\end{abstract}

\section{Introduction}

Two fundamental problems that have been extensively studied separately by 
different research communities for many years are the $\ksum{k}$ problem and the 
problem of finding subgraphs of a certain form in a graph. We investigate the 
relationships between these problems and show tight connections between 
$\ksum{k}$ and the ``exact-weight'' version of the subgraph finding problem.

The $\ksum{k}$ problem is the parameterized version of the well known 
NP-complete problem \text{\scshape subset-sum}, and it asks if in a set of $n$ 
integers, there is a subset of size $k$ whose integers sum to $0$. This problem 
can be solved easily in time $O(n^{\ctwo{k}})$, and Baran, Demaine, and 
$\patrascu$~\cite{BDP08} show how the $\ksum{3}$ problem can be solved in time 
$O(n^2 / \log^2 n)$ using certain hashing techniques.
However, it has been a longstanding open problem to solve $\ksum{k}$ for 
\emph{some} $k$ in time $O(n^{\ctwo{k} - \epsilon})$ for some $\epsilon > 0$.
 In certain restricted models of computation, an $\Omega(n^{\ctwo{k}})$ lower 
 bound has been established initially by Erickson~\cite{Eri95} and later 
 generalized by Ailon and Chazelle~\cite{AC05}, and recently, $\patrascu$ and 
 Williams~\cite{PW10} show that $n^{o(k)}$ time algorithms for all $k$ would 
 refute the Exponential Time Hypothesis.
 The literature seems to suggest the following hypothesis, which we call the $\ksumconj{k}$:
 
\begin{conjecture}[The $\ksumconj{k}$]
	There does not exist a $k \geq 2$, an $\varepsilon > 0$, and a randomized 
	algorithm that succeeds (with high probability) in solving $\ksum{k}$ in time 
	$O(n^{\cftwo{k}-\varepsilon})$.
\end{conjecture}
 
The presumed difficulty of solving $\ksum{k}$ in time 
$O(n^{\ctwo{k}-\varepsilon})$ for any $\varepsilon>0$ has been the 
basis of many conditional lower bounds for problems in computational geometry. 
The $k=3$ case has re
ceived even more attention, and proving $\ksum{3}$-hardness 
has become common practice in the computational geometry literature. In a recent 
line of work, $\patrascu$ \cite{Pat10}, Vassilevska and Williams \cite{VW09}, 
and Jafargholi and Viola \cite{JV13} show conditional hardness based on 
$\ksum{3}$ for problems in data structures and triangle problems in graphs.


The problem of determining whether a weighted or unweighted $n$-node graph has a 
subgraph that is isomorphic to a fixed $k$ node graph $H$ with some properties 
has been well-studied in the past \cite{NP85,KKM00,EG04}. There has been much 
work on detection and counting copies of $H$ in graphs, the problem of listing 
all such copies of $H$, finding the minimum-weight copy of $H$, etc. 
\cite{VW09,KLL11}. Considering these problems for restricted types of subgraphs 
$H$ has received further attention, such as for subgraphs $H$ with large 
indepedent sets, or with bounded treewidth, and various other structures 
\cite{VW09,Wil09,KLL11,FLRSR12}. In this work, we focus on the following 
subgraph finding problem.

\begin{definition}[The ${\eweight{H}}$ Problem]
	Given an edge-weighted graph $G$, does there exist a (not necessarily induced) 
	subgraph isomorphic to $H$ such that the sum of its edge weights equals a 
	given target value $t$?\footnote{We can assume, without loss of generality, 
	that the target value is always $0$ and that $H$ has no isolated vertices.}
\end{definition}

No non-trivial algorithms were known for this problem. Theoretical evidence for 
the hardness of this problem was given in \cite{VW09}, where the authors prove 
that for any $H$ of size $k$, an algorithm for the exact-weight problem can give 
an algorithm for the minimum-weight problem with an overhead that is only $O(2^k 
\cdot \log M)$, when the weights of the edges are integers in the range 
$[-M,M]$. They also show that improving on the trivial $O(n^3)$ upper bound for $\eweight{\triangle}$ to $O(n^{3-\varepsilon})$ for any $\varepsilon>0$ would not only imply an $\Ot(n^{3-\varepsilon})$ algorithm\footnote{In our bounds, $k$ is treated as a constant. The notation $\Ot(f(n))$ will hide $\polylog(n)$ factors.}
 for the minimum-weight $\kclique{3}$ problem, which from \cite{WW10} 
is in turn known to imply faster algorithms for the canonical $\apsp$ problem, but also 
an $O(n^{2-\varepsilon'})$ upper bound for the $\ksum{3}$ problem, for some $\varepsilon'>0$. They give additional 
evidence for the hardness of the exact-weight problem by proving that faster 
than trivial algorithms for the $\kclique{k}$ problem will break certain 
cryptographic assumptions.

Aside from the aforementioned reduction from $\ksum{3}$ to 
$\eweight{\triangle}$, few other connections between $\ksum{k}$ and subgraph 
problems were known. The standard reduction from Downey and Fellows~\cite{DF95} 
gives a way to reduce the unweighted $\kclique{k}$ detection problem to 
$\ksum{\binom{k}{2}}$ on $n^2$ numbers. Also, in \cite{Pat10} and  \cite{JV13}, 
strong connections were shown between the $\ksum{3}$ problem (or, the similar 
$\kxor{3}$ problem) and listing triangles in unweighted graphs.


\subsection{Our Results}

In this work, we study the exact-weight subgraph problem and its connections to 
$\ksum{k}$. We show three types of reductions: $\ksum{k}$ to subgraph problems, 
subgraphs to other subgraphs, and subgraphs to $\ksum{k}$. These results give 
conditional lower bounds that can be viewed as showing hardness either for 
$\ksum{k}$ or for the subgraph problems. We focus on showing implications of the 
$\ksumconj{k}$ and therefore view the first two kinds as a source for 
conditional lower bounds for $\eweight{H}$, while we view the last kind as 
algorithms for solving the problem. Our results are summarized in Table~1 and 
Figure~1, and are discussed below.

\paragraph{Hardness.} By embedding the numbers of the $\ksum{k}$ problem into 
the edge weights of the exact-weight subgraph problem, using different encodings 
depending on the structure of the subgraph, we prove four reductions that are 
summarized in Theorem~\ref{thm1}:
\begin{theorem} \label{thm1}
Let $k\geq 3$. If for all $\varepsilon>0$, $\ksum{k}$ cannot be solved in time 
$O(n^{\ctwo{k} - \varepsilon})$, then none of the following 
problems\footnote{$\kmatching{k}$ denotes the $k$-edge matching on $2k$ nodes. 
$\kstar{k}$ denotes the $k$-edge star on $k+1$ nodes. $\kpath{k}$ denotes the 
$k$-node path on $k$-$1$ edges.} can be solved in time $O(n^{\ctwo{k}-\delta})$, 
for any $\delta>0$:
\begin{itemize}
\item $\eweight{H}$ on a graph on $n$ nodes, for any subgraph $H$ on $k$ nodes.
\item $\eweight{\kmatching{k}}$ on a graph on $\sqrt{n}$ nodes.
\item $\eweight{\kstar{k}}$ on a graph on $n^{(1-1/k)}$ nodes.
\item $\eweight{\kmpath}$ on a graph on $n$ nodes.
\end{itemize}
\end{theorem}

An immediate implication of Theorem~\ref{thm1} is that neither $\kstar{3}$ 
can be solved in time $O(n^{3-\varepsilon})$, nor can $\kmatching{3}$ be solved in time
$O(n^{4-\varepsilon})$ for some $\varepsilon>0$, unless $\ksum{3}$ can be solved in time $O(n^{2-\varepsilon'})$ for some $\varepsilon'>0$. We later show 
that an $O(n^{2-\varepsilon})$ algorithm for $\ksum{3}$ for some $\varepsilon>0$ will imply both an $\Ot(n^{3-\varepsilon})$ 
algorithm for $\kstar{3}$ and an  $\Ot(n^{4-2\varepsilon})$ algorithm for $\kmatching{3}$. In 
other words, either \emph{all} of the following three statements are true, or 
\emph{none} of them are:
\begin{itemize}
	\item $\ksum{3}$ can be solved in time $O(n^{2-\varepsilon})$ for some $\varepsilon>0$.
	\item $\kstar{3}$ can be solved in time $O(n^{3-\varepsilon})$ for some $\varepsilon>0$.
	\item $\kmatching{3}$ can be solved in time $O(n^{4-\varepsilon})$ for some $\varepsilon>0$.
\end{itemize}

From \cite{VW09}, we already know that solving $\triangle$ in time $O(n^{3-\varepsilon})$ for some $\varepsilon>0$ 
implies that $\ksum{3}$ can be solved in time $O(n^{2-\varepsilon})$ for some $\varepsilon>0$. By Theorem~\ref{thm1}, 
this would imply faster algorithms for $\kstar{3}$ and $\kmatching{3}$ as well.

Another corollary of Theorem~\ref{thm1} is the fact that $\kpath{4}$ cannot be 
solved in time $O(n^{3-\varepsilon})$ for some $\varepsilon>0$ unless $\ksum{5}$ can be solved in time $O(n^{3-\varepsilon'})$ for some $\varepsilon'>0$. 
This is in sharp contrast to the unweighted version (and the min-weight version) 
of $\kpath{4}$, which can both be solved easily in time $O(n^2)$.

Theorem~\ref{thm1} shows that the $\ksum{k}$ problem can be reduced to the 
$\eweight{H}$ problem for various types of subgraphs $H$, and as we noted, this 
implies connections between the $\ews$ problem for different subgraphs. It is 
natural to ask if for any other subgraphs the $\ews$ problems can be related to 
one another. We will answer this question in the affirmative---in particular, we 
show a tight reduction from $\triangle$ to $\kpath{4}$.

To get this result, we use the edge weights to encode information about the 
nodes in order to prove a reduction from $\eweight{H_1}$ to $\eweight{H_2}$, 
where $H_1$ is what we refer to as a ``$\vertminor$'' of $H_2$. Informally, a 
$\vertminor$ of a graph is one that is obtained by edge deletions and node 
identifications (contractions) for arbitrary pairs of nodes of the original graph (see 
Section~\ref{sec:vm} for a formal definition). For example, the triangle 
subgraph is a $\vertminor$ of the path on four nodes, which is itself a 
$\vertminor$ of the cycle on four nodes.

\begin{theorem} \label{thm2}
	Let $H_1,H_2$ be subgraphs such that $H_1$ is a $\vertminor$ of $H_2$. For any 
	$\alpha \geq 2$, if $\eweight{H_2}$ can be solved in time $O(n^\alpha)$, then 
	$\eweight{H_1}$ can be solved in time $\tilde{O}(n^\alpha)$.
\end{theorem}

Therefore, Theorem~\ref{thm2} allows us to conclude that $\kcycle{4}$ cannot be 
solved in time $O(n^{3-\varepsilon})$ for some $\varepsilon>0$ unless $\kpath{4}$ can be solved in time $\Ot(n^{3-\varepsilon})$, 
which cannot happen unless $\triangle$ can be solved in time $\Ot(n^{3-\varepsilon})$. 

To complete the picture of relations between $3$-edge subgraphs, consider the 
subgraph composed of a 2-edge path along with another (disconnected) edge. We 
call this the ``$\VI$'' subgraph and we define the $\eweight{\VI}$ problem 
appropriately. Since the path on four nodes is a $\vertminor$ of $\VI$, we have 
that an $O(n^{3-\varepsilon})$ for some $\varepsilon>0$ algorithm for $\eweight{\VI}$ implies an $\Ot(n^{3-\varepsilon})$ algorithm 
for $\kpath{4}$. In Figure~1, we show this web of connections between the 
exact-weight $3$-edge subgraph problems and its connection to $\ksum{3}$, 
$\ksum{5}$, and $\apsp$. In fact, we will soon see that the conditional lower 
bounds we have established for these $3$-edge subgraph problems are all tight. 
Note that the detection and minimum-weight versions of some of these $3$-edge subgraph 
problems can all be solved much faster than $O(n^3)$ (in particular, $O(n^2)$), and yet 
such an algorithm for the exact-weight versions for \emph{any} of these problems 
will refute the $\ksum{3}$ Conjecture, the $\ksum{5}$ Conjecture, and lead to 
breakthrough improvements in algorithms for solving $\apsp$ and other important 
graph and matrix optimization problems (cf. \cite{WW10})!

Another $O(n^3)$ solvable problem is the $\eweight{\kpath{5}}$, and by noting 
that both $\kcycle{4}$ and $\VI$ are $\vertminor$s of $\kpath{5}$, we get that 
improved algorithms for $\kpath{5}$ will yield faster algorithms for all of the 
above problems. Moreover, from Theorem~\ref{thm1}, $\ksum{6}$ reduces to 
$\kpath{5}$. This established $\eweight{\kpath{5}}$ as the ``hardest'' of the $O(n^3)$ 
time problems that we consider.



We also note that Theorem~\ref{thm2} yields some interesting consequences under 
the assumption that the $\kclique{k}$ problem cannot be solved in time 
$O(n^{k-\varepsilon})$ for some $\varepsilon>0$. Theoretical evidence for this assumption was provided in 
\cite{VW09}, where they show how an $O(n^{k-\varepsilon})$ for some $\varepsilon>0$ time algorithm for 
$\eweight{\kclique{k}}$ yields a sub-exponential time algorithm for the 
multivariate quadratic equations problem, a problem whose hardness is assumed in 
post-quantum cryptography.

We note that the $4$-clique is a $\vertminor$ of the $8$-node path, and so by 
Theorem~\ref{thm2}, an $O(n^{4-\varepsilon})$ for some $\varepsilon>0$ algorithm for $\kpath{8}$ will yield a faster 
$\kclique{4}$ algorithm. Note that an $O(n^{5-\varepsilon})$ algorithm for $\kpath{8}$ 
already refutes the $\ksum{9}$ Conjecture. However, this by itself is not known 
to imply faster clique algorithms\footnote{It is not known whether the 
assumption that $\kclique{k}$ cannot be solved in time $O(n^{k-\varepsilon})$ for any $\varepsilon>0$ is stronger or 
weaker than the $\ksum{k}$ Conjecture.}. Also, there are other subgraphs for 
which one can only rule out $O(n^{\ctwo{k}-\varepsilon})$ for $\varepsilon>0$ upper bounds from the 
$\ksum{k}$ Conjecture, while assuming hardness for the $\kclique{k}$ problem and 
using Theorem~\ref{thm2}, much stronger lower bounds can be achieved.


\begin{figure}
	\label{fig:reductions}
	\centering
	\includegraphics[scale=0.8]{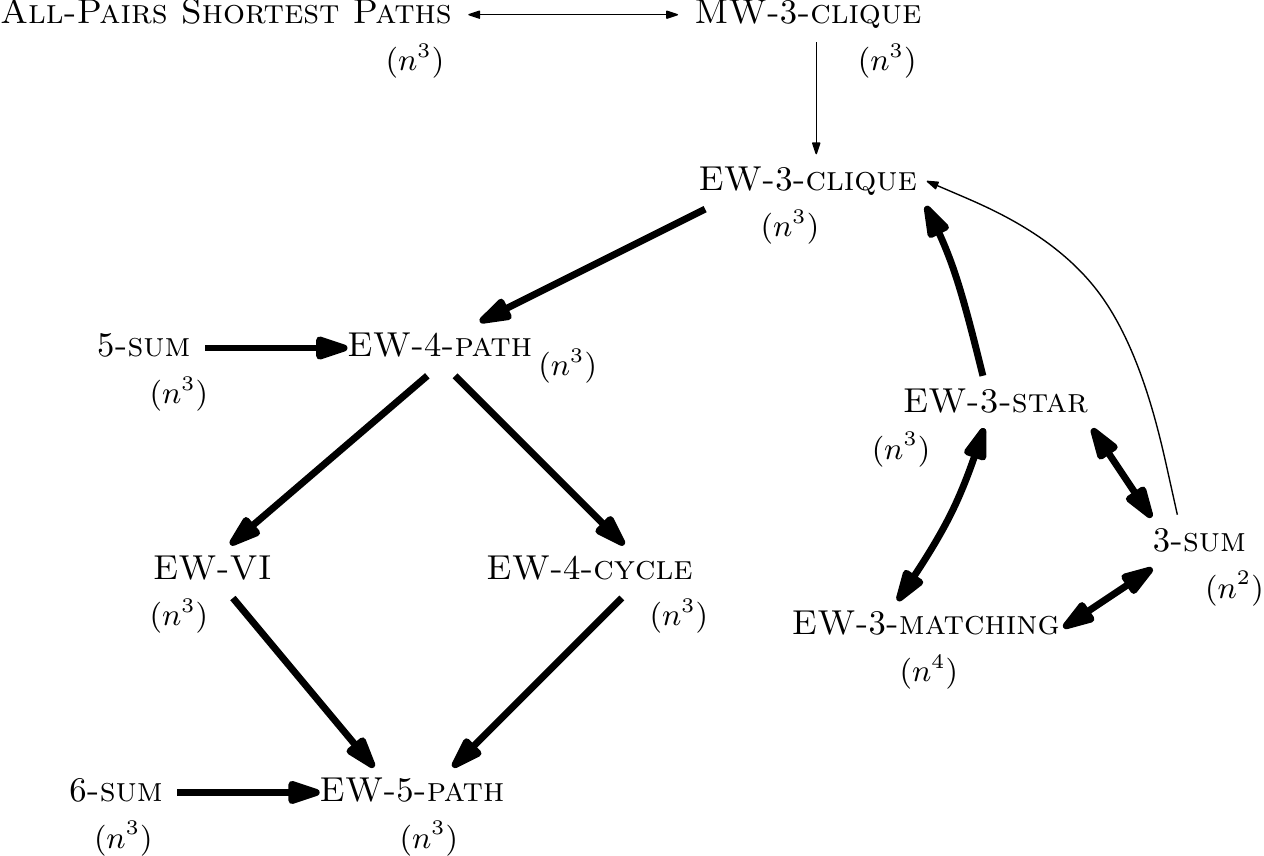}
	\caption{A diagram of the relationships between $\eweight{H}$ (denoted 
	${\textsc EW}$, for small subgraphs $H$) and other important problems. The 
	best known running times are given for each problem, and an arrow $A \to B$ 
	denotes that $A$ can be tightly reduced to $B$, in the sense that improving 
	the stated running time for $B$ will imply an improvement on the stated 
	running time for $A$. The reductions established in this work are displayed in 
	bold, the others are due to \cite{WW10}, \cite{VW09}. }
\end{figure}

\paragraph{Algorithms.} So far, our reductions only show one direction of the 
relationship between $\ksum{k}$ and the exact-weight subgraph problems. We now 
show how to \emph{use} $\ksum{k}$ to solve $\eweight{H}$, which will imply that 
many of our previous reductions are indeed tight. The technique for finding an 
$H$-subgraph is to enumerate over a set of $d$ smaller subgraphs that partition 
$H$ in a certain way. Then,  in order to determine whether the weights of these 
$d$ smaller subgraphs sum up to the target weight, we use $\ksum{d}$.
We say that $(S, H_1, \ldots, H_d)$ is a \emph{$d$-separator} of $H$ iff $S, 
H_1, \ldots, H_d$ partition $V(H)$ and there are no edges between a vertex in 
$H_i$ and a vertex in $H_j$ for any distinct $i,j \in [d]$.

\begin{theorem} \label{thm3}
	Let $(S, H_1, \ldots, H_d)$ be a $d$-separator of $H$. Then, $\eweight{H}$ can 
	be reduced to $\Ot(n^{|S|})$ instances of $\ksum{d}$ each on $\max\{n^{|H_1|}, 
	\ldots, n^{|H_d|} \}$ numbers.
\end{theorem}

By using the known $\ksum{d}$ algorithms, Theorem~\ref{thm3} gives a non-trivial 
algorithm for exact-weight subgraph problems. The running time of the algorithm 
depends on the choice of the separator used for the reduction. We observe that 
the optimal running time can be achieved even when $d=2$, and can be identified 
(naively, in time $O(3^k)$) using the following expression. Let
\[ \gamma(H) = \min_{(S, H_1, H_2) \text{ is a $2$-separator}} \left\{ |S| + 
\max \left\{ |H_1|, |H_2| \right \} \right\}. \]

\begin{corollary} \label{cor4}
	$\eweight{H}$ can be solved in time $\Ot(n^{\gamma(H)})$.
\end{corollary}

Corollary~\ref{cor4} yields the upper bounds that we claim in Figure~1 and 
Table~1. For example, to achieve the $O(n^{\ctwo{(k+1)}})$ time complexity for 
$\kpath{k}$, observe that we can choose the set containing just the ``middle'' 
node of the path to be $S$, so that the graph $H \setminus S$ is split into two 
disconnected halves $H_1$ and $H_2$, each of size at most $\ctwo{(k-1)}$. Note 
that this is the optimal choice of a separator, and so $\gamma(\kpath{k}) = 1 + 
\ctwo{(k-1)} = \ctwo{(k+1)}$.
It is interesting to note that this simple algorithm achieves running times that 
match many of our conditional lower bounds. This means that in many cases, 
improving on this algorithm will refute the $\ksum{k}$ Conjecture, and in fact, 
we are not aware of any subgraph for which a better running time is known. 

$\eweight{H}$ is solved most efficiently by our algorithm when $\gamma(H)$ is 
small, that is, subgraphs with small ``balanced'' separators. Two such cases are 
when $H$ has a large independent set and when $H$ has bounded treewidth. We show 
that $\eweight{H}$ can be solved in time $O(n^{k-\fftwo{s}})$, if $\alpha(H)=s$, 
and in time $O(n^{\frac{2}{3}\cdot k + tw(H)})$. Also, we observe that our 
algorithm can be modified slightly to get an algorithm for the minimization 
problem.

\begin{theorem}
\label{thm4}
	Let $H$ be a subgraph on $k$ nodes, with independent set of size $s$. Given a 
	graph $G$ on $n$ nodes with node and edge weights, the minimum total weight of 
	a (not necessarily induced) subgraph of $G$ that is isomorphic to $H$ can be 
	found in time $\Ot(n^{k-s+1})$.
\end{theorem}

This algorithm improves on the $O(n^{k-s+2})$ time algorithm of Vassilevska and 
Williams \cite{VW09} for the \text{\sc Min-Weight-H} problem.


\begin{table}[ht]
	{ \centering \small
\begin{tabular}{@{}lllllllllll@{}} \toprule

Subgraph & \vspacer & Exact & \vspacer & Lower Bound & \vspacer & Condition & 
\vspacer & Detection & \vspacer & Min \\

\midrule

$\triangle$ && $n^3$ && $n^3$ && $\ksum{3}$, APSP && $n^{\omega}$ \cite{IR77} && 
$n^3$ \\

$\kpath{4}$ && $n^3$ && $n^3$ && $\triangle$, $\ksum{5}$ && $n^2$ && $n^2$ \\

\midrule

$\kmatching{k}$ && $n^{2\cdot \cftwo{k}}$ && $n^{2\cdot \cftwo{k}}$ && 
$\ksum{k}$ && $n^2$ && $n^2$ \\

\mrany{$\kstar{k}$} && \mrany{$n^{\cftwo{k}+1}$} && $n^{\cftwo{k+1}}$ && 
$\kpsum$ && \mrany{$n$} && \mrany{$n^2$} \\
 
&& && $n^{\cftwo{k}\cdot\frac{k}{k-1}}$ && $\ksum{k}$ && && \\

$\kpath{k}$ && $n^{\cftwo{k+1}}$ && $n^{\cftwo{k+1}}$ && $\ksum{k}$ && $n^2$ && 
$n^2$ \\

$\kcycle{k}$ && $n^{\cftwo{k}+1}$ && $n^{\cftwo{k+1}}$ && $\kpath{k}$ && $n^2$ 
&& $n^3$ \\

\midrule

\mrany{Any} && \mrany{$n^k$} && $n^{\ctwo{k}}$ && $\ksum{k}$ && 
\mrany{$n^{\omega k/3}$ \cite{KKM00}} && \mrany{$n^k$} \\

&& && $n^{\varepsilon k}$ && (ETH) && && \\

$\alpha(H) = s$ && $n^{k-\fftwo{s}}$ && $n^{\cftwo{k}}$ && $\ksum{k}$ && 
$n^{k-s+1}$ \cite{KLL11} && $n^{k-s+1}$ [Thm. 4] \\

$tw(H) = w$ && $n^{\frac{2}{3}k + w}$ && $n^{\cftwo{k}}$ && $\ksum{k}$ && 
$n^{w+1}$ \cite{AYZ95} && $n^{2w}$ \cite{FLRSR12} \\

\midrule

$\ksum{k}$ && $n^{\cftwo{k}}$ && $n^{\cftwo{k}}$ && $\kmatching{k}$, $\kstar{k}$ 
&& - && - \\

\bottomrule
\end{tabular} \vspace{1em}
\label{table1} \\
\caption{The results shown in this work for $\eweight{H}$ for various $H$. The 
second column has the upper bound achieved by our algorithm from 
Corollary~\ref{cor4}. Improvements on the lower bound in the third column will 
imply improvements on the best known algorithms for the problems in the 
condition column. These lower bounds are obtained by our redctions, except for 
the first row which was proved in \cite{VW09}. For comparison, we give the 
running times for the (unweighted) detection and minimum-weight versions of the 
subgraph problems. The last row shows our conditional lower bounds for 
$\ksum{k}$. $\alpha(H)$ represents the independence number of $H$, $tw(H)$ is 
its treewidth. The results for the  ``Any'' row hold for all subgraphs on $k$ 
nodes. ETH stands for the Exponential Time Hypothesis. }

}
\end{table}

\paragraph{Organization.} We give formal definitions and preliminary reductions in 
Section~\ref{sec:prelim}.
In Section~\ref{sec:lowerbounds} we present reductions from $\ksum{k}$ to $\ews$ 
subgraph problems that prove Theorem~\ref{thm1}.
In Section~\ref{sec:vm} we define $\vertminor$s and prove Theorem~\ref{thm2}.
In Section~\ref{sec:algs}, we give reductions to $\ksum{k}$ and prove Theorems~\ref{thm3} and~\ref{thm4}.


\section{Preliminaries and Basic Constructions}
\label{sec:prelim}

For a graph $G$, we will use $V(G)$ to represent the set of vertices and $E(G)$ 
to represent the set of edges. The notation $N(v)$ will be used to represent the 
neighborhood of a vertex $v \in V(G)$.

\subsection{Reducibility}

We will use the following notion of reducibility between two problems. In 
weighted graph problems where the weights are integers in $[-M,M]$, $n$ will 
refer to the number of nodes times $\log M$. For $\ksum{k}$ problems where the 
input integers are in $[-M,M]$, $n$ will refer to the number of integers times 
$\log M$. In Appendix~\ref{app:reducibility} we formally define our notion of 
reducibility, which follows the definition of subcubic reductions in 
\cite{WW10}. Informally, for any two decision problems $A$ and $B$, we say that 
$A \rto{a}{b} B$ if for any $\epsilon > 0$, there exists a $\delta > 0$ such 
that if $B$ can be solved (w.h.p.) in time $n^{b - \epsilon}$, then $A$ can be 
solved (w.h.p.) in time $O(n^{a - \delta})$, where $n$ is the size of the input. 
Note that $\polylog(n)$ factor improvements in solving $B$ may not imply any 
improvements in solving $A$. Also, we say that $A \equivto{a}{b} B$ if and only 
if $A \rto{a}{b} B$ and $B \rto{b}{a} A$.

\subsection{The $\ksum{k}$ Problem}

Throughout the paper, it will be more convenient to work with a version of the 
$\ksum{k}$ problem that is more structured than the basic formulation. This 
version is usually referred to as either $\text{\scshape table-\text{$k$}-sum}$ 
or $\ksum{k}'$, and is known to be equivalent to the basic formulation, up to 
$k^k$ factors (by a simple extension of Theorem 3.1 in \cite{GO95}). For 
convenience, and since $f(k)$ factors are ignored in our running times, we will 
refer to this problem as $\ksum{k}$.

\begin{definition}[${\tableksum{k}}$]
	Given $k$ lists $L_1,\ldots,L_k$ each with $n$ numbers where $L_i = \{ x_{i,j} 
	\}_{j\in [n]} \subseteq \mathbb{Z}$, do there exist $k$ numbers 
	$x_{1,a_1},\ldots, x_{k,a_k}$, one from each list, such that $\sum_{i=1}^k 
	x_{i,a_i} = 0$?
\end{definition}

In our proofs, we always denote an instance of $\tableksum{k}$ by $L_1, \ldots, 
L_k$, where $L_i = \{ x_{i,j} \}_{j\in [n]} \subseteq \mathbb{Z}$, so that
$x_{i,j}$ is the $\jth$ number of the $\ith$ list $L_i$. We define a $\ksol{k}$ 
to be a set of $k$ numbers $\ksolsia{k}$, one from each list. The sum of a 
$\ksol{k}$ $\ksolsia{k}$ will be defined naturally as $\sum_{i=1}^k 
x_{i,a_i}$.

In \cite{Pat10}, $\patrascu$ defines the $\convksum{3}$ problem. We consider a 
natural extension of this problem.

\begin{definition}[${\convksum{k}}$]
	Given $k$ lists $L_1, \ldots, L_k$ each with $n$ numbers, where $L_i = \{ 
	x_{i,j} \}_{j\in [n]} \subseteq \mathbb{Z}$, does there exist a $\ksol{k}$ 
	$\ksolsia{k}$ such that $a_k = a_1 + \cdots + a_{k-1}$ and $\sum_{i=1}^k 
	x_{i,a_i} = 0$?
\end{definition}

Theorem~10 in \cite{Pat10} shows that $\ksum{3} \rto{2}{2} \convksum{3}$. By 
generalizing the proof, we show the following useful lemma (see proof in 
Appendix~\ref{app:conv}).

\begin{lemma}
	\label{lem:convksum}
	For all $k \geq 2$, $\ksum{k} \rto{\ctwo{k}}{\ctwo{k}} \convksum{k}$.
\end{lemma}



\subsection{$H$-Partite Graphs}

Let $H$ be a subgraph on $k$ nodes with $V(H) = \{h_1, \ldots, h_k\}$.

\begin{definition}[${H}$-partite graph]
	Let $G$ be a graph such that $V(G)$ can be partitioned into $k$ sets 
	$\P_{h_1}, \ldots, \P_{h_k}$, each containing $n$ vertices. We will refer to 
	these $k$ sets as the \emph{super-nodes} of $G$. A pair of super-nodes 
	$(\P_{h_i}, \P_{h_j})$ will be called a \emph{super-edge} if $(h_i,h_j) \in 
	E(H)$.	Then, we say that $G$ is \emph{$H$-partite} if every edge in $E(G)$ 
	lies in some super-edge of $G$.
\end{definition}

We denote the set of vertices of an $H$-partite graph $G$ by $V(G) = \{ v_{i,j} 
\}_{i\in [k], j\in [n]}$, where $v_{i,j}$ is the $j^{th}$ vertex in super-node 
$\P_{h_i}$. We will say that $G$ is the \emph{complete $H$-partite graph} when 
$(v_{i,a}, v_{j,b}) \in E(G)$ if and only if $(\P_{h_i}, \P_{h_j})$ is a 
super-edge of $G$, for all $a,b \in [n]$.

An \emph{$H$-subgraph} of an $\hpg{H}$ $G$, denoted by $\chi = \hsub{i}{a}{k} 
\subseteq V(G)$, is a set of vertices for which there is exactly one vertex 
$\hsubv{i}{a}$ from each $\sn$ $\P_{h_i}$, where $a_i$ is an index in $[n]$. 
Given a weight function $w: (V(G) \cup E(G)) \rightarrow \mathbb{Z}$ for the 
nodes and edges of $G$, the total weight of the subgraph $\chi$ is defined 
naturally as

\[ w(\chi) = \sum\limits_{h_i\in V(H)} w(\hsubv{i}{a}) + 
\sum\limits_{(h_i,h_j)\in E(H)} w(\hsubv{i}{a}, \hsubv{j}{a}). \] 

Figure~\ref{fig:hsubgraph} illustrates our definitions and notations of 
$\hpg{H}$s and $H$-subgraphs.

\begin{figure}
	\begin{minipage}[b]{0.5\linewidth}
		\centering
		\subfigure[A subgraph $H$]{
			\includegraphics[scale=0.4]{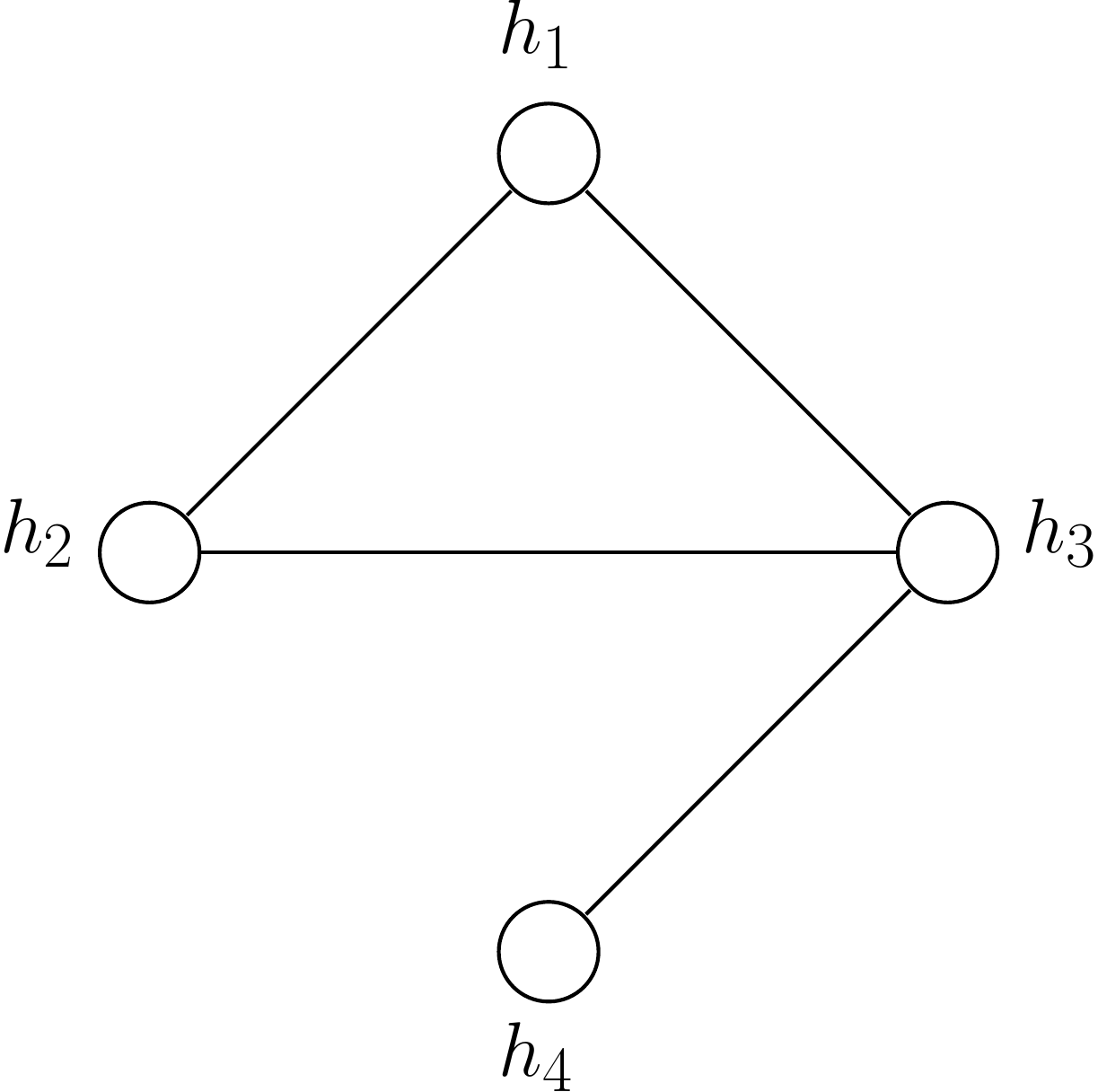}
			\label{subfig:original}
		}
	\end{minipage}
	\begin{minipage}[b]{0.5\linewidth}
		\centering
		\subfigure[An $H$-partite graph $G$]{
			\includegraphics[scale=0.23]{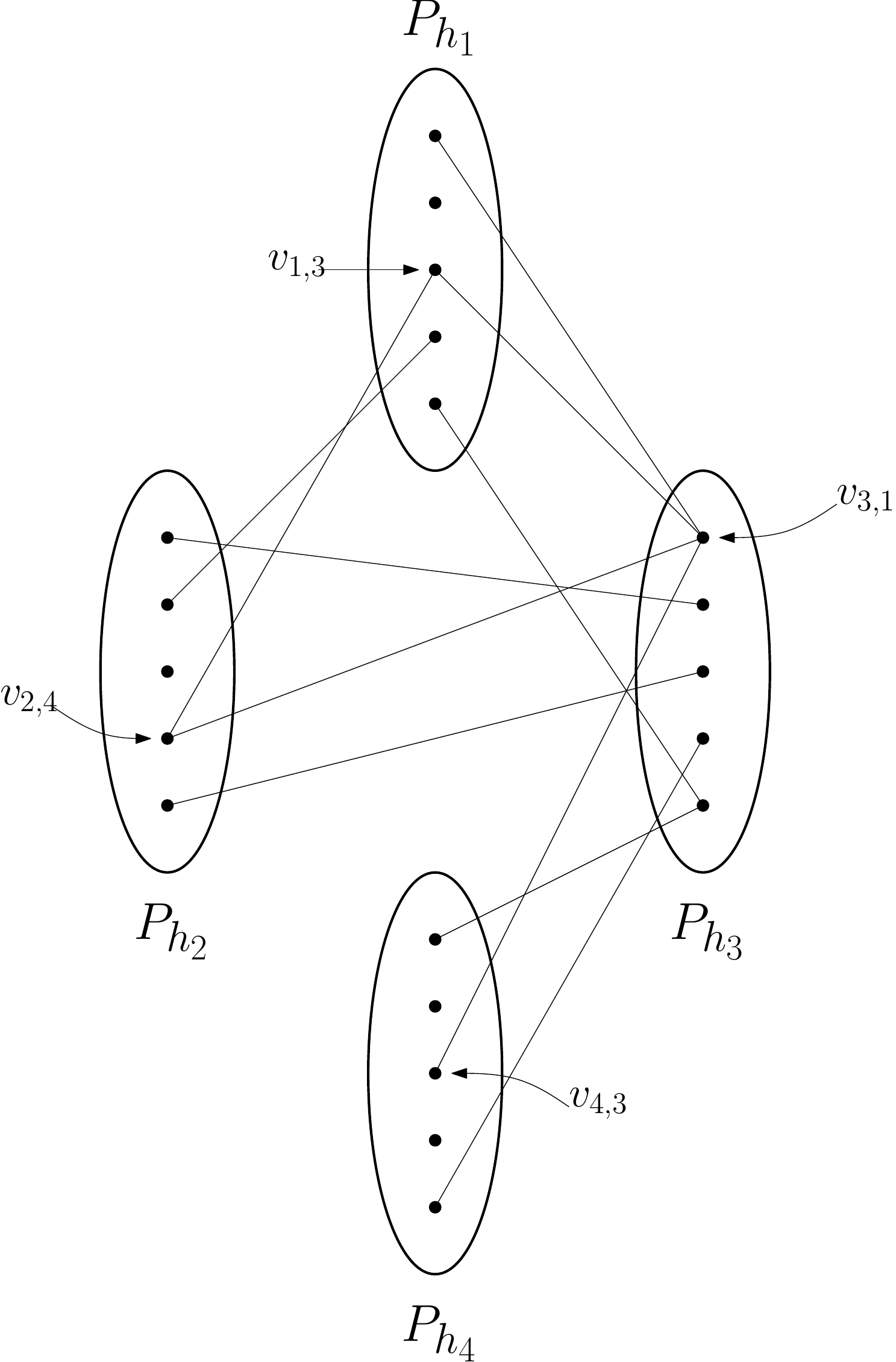}
			\label{subfig:3core}
		}
	\end{minipage}

	\caption{A subgraph $H$ along with an $H$-partite graph $G$. We associate a 
	partition $\P_{h_i} \subseteq V(G)$ with each node $h_i \in V(H)$. The vertex 
	$v_{i,j} \in V(G)$ represents the $j^{th}$ node in partition $\P_{h_i}$ in 
	$G$. Note that the set $\{ v_{1,3}, v_{2,4}, v_{3,1}, v_{4,3} \}$ is an 
	$H$-subgraph in $G$. Also, since there are no edges between a vertex in 
	$\P_{h_3}$ and a vertex in $\P_{h_4}$, $G$ is a valid $H$-partite graph.}
	\label{fig:hsubgraph}
\end{figure}

Now, we define a more structured version of the $\eweight{H}$ problem which is 
easier to work with.

\begin{definition}[The ${\ew{H}}$ Problem]
	Given a complete $\hpg{H}$ graph $G$ with a weight function $w: (V(G)\cup 
	E(G))\rightarrow \mathbb{Z}$ for the nodes and edges, does there exist an 
	$H$-subgraph of total weight $0$?
\end{definition}

In appendix~\ref{app:hpg}, we prove the following lemma, showing that the two 
versions of the $\eweight{H}$ problem are reducible to one another in a tight 
manner. All of our proofs will use the formulation of $\ew{H}$, yet the results 
will also apply to $\eweight{H}$. Note that our definitions of $H$-partite 
graphs uses ideas similar to color-coding \cite{AYZ95}.

\begin{lemma}
	\label{lem:ewhequiv}
	Let $\alpha > 1$. $\eweight{H} \equivto{\alpha}{\alpha} \ew{H}$.
\end{lemma}

\section{Reductions from $\ksum{k}$ to Subgraph Problems} 
\label{sec:lowerbounds}


In this section we prove Theorem~\ref{thm1} by proving four reductions, each of 
these reductions uses a somewhat different way to encode $\ksum{k}$ in the 
structure of the subgraph. First, we give a generic reduction from $\ksum{k}$ to 
$\ew{H}$ for an arbitrary $H$ on $k$ nodes. We set the node weights of the graph 
to be the numbers in the $\ksum{k}$ instance, in a certain way. 

\begin{lemma}[${\ksum{k} \rto{\ctwo{k}}{\ctwo{k}} \ew{H}}$]
	\label{lem:generalbound}
	Let $H$ be a subgraph with $k$ nodes. Then, $\ksum{k}$ on $n$ numbers can be 
	reduced to a single instance of $\ew{H}$ on $kn$ vertices.
\end{lemma}

\begin{proof}
Let $H$ be a subgraph with node set $\{h_1,\ldots,h_k\}$. Given a 
$\tableksum{k}$ instance of $k$ lists $L_1, \ldots, L_k$, where each $L_i = 
\{x_{i,j} \}_{j \in [n]}$, we create a complete $H$-partite graph $G$ on $kn$ 
vertices where we associate each super-node $\P_{h_i}$ with a list $L_i$, and 
the $\jth$ vertex in the $\sn$ $v_{i,j}$ with the number $x_{i,j} \in L_i$. To 
do this, we set all edge weights to be $0$, and for every $i\in [k], j\in [n]$, 
we set $w(v_{i,j})=x_{i,j}$. Now, for any $H$-subgraph $\chi = \hsub{i}{a}{k}$ 
of $G$, the total weight of $\chi$ will be exactly $\sum_{i=1}^k x_{i,a_i}$, 
which is the sum of the $\ksol{k}$ $\ksolsia{k}$. For the other direction, for 
any $\ksol{k}$ $\ksolsia{k}$, the $H$-subgraph $\hsub{i}{a}{k}$ of $G$ has 
weight exactly $\sum_{i=1}^k x_{i,a_i}$. Therefore, there is a $\ksol{k}$ of sum 
$0$ iff there is an $H$-subgraph in $G$ of total weight $0$. 
\end{proof}

We utilize the edge weights of the graph, rather than the node weights, to prove 
a tight reduction to $\kmatching{k}$. 

\begin{lemma}[${\ksum{k} \rto{\ctwo{k}}{2 \cdot \ctwo{k}} 
	\ew{\kmatching{k}}}$]
	\label{lem:matching}
	Let $H$ be the $\kmatching{k}$ subgraph. Then, $\ksum{k}$ on $n$ numbers can 
	be reduced to a single instance of $\ew{H}$ on $k \sqrt{n}$ vertices.
\end{lemma}

\begin{proof}
Given $k$ lists $L_1,\ldots,L_k$ each with $n$ numbers, we will construct a 
complete $\hpg{H}$ $G$ on $k \sqrt{n}$ vertices, where there will be super-edges 
$(\P_{h_{2i-1}}, \P_{h_{2i}})$ for each $i \in [k]$, each with $n$ edges over 
$\sqrt{n}$ vertices. We place each number in $L_i$ on an arbitrary edge within 
the $i^{th}$ $\se$ of $G$ by setting $w(v_{2i-1,a},v_{2i,b}) = x_{i,(a\cdot 
\sqrt{n} + b)}$ for all $a,b \in [\sqrt{n}]$. Now, note that the $H$-subgraph 
$\hsub{i}{a}{2k}$ of $G$ has weight $\sum_{i=1}^k x_{i,b_i}$, where $b_i = 
a_{2i-1} \cdot \sqrt{n} + a_{2i}$. This is precisely the sum of the $\ksol{k}$ 
$\ksols{i}{b}{k}$. And, for every $\ksol{k}$ $\ksols{i}{b}{k}$, if we choose 
$a_{2i-1}=\lfloor b_i / \sqrt{n} \rfloor$ and $a_{2i} = b_i - a_{2i-1} \sqrt{n} 
$, the $H$-subgraph $\hsub{i}{a}{2k}$ has weight $\sum_{i=1}^k x_{i,b_i}$. 
Therefore, there is a $\ksol{k}$ of sum $0$ iff there is an $H$-subgraph in $G$ 
of total weight 0. 
\end{proof}

Another special type of subgraph which can be shown to be tightly related to the 
$\ksum{k}$ problem is the $k$-edge star subgraph. We define the $\kstar{k}$ subgraph $H$ to be such that $V(H)=\{h_1,\ldots,h_k,h_{k+1} \}$ and $E(H)= \{ (h_1, h_{k+1}), (h_2, h_{k+1}), 
\ldots, (h_k, h_{k+1}) \}$, so that $h_{k+1}$ is the center node.

\begin{lemma}[${\ksum{k} \rto{\cftwo{k}}{\cftwo{k}\cdot\frac{k}{k-1}} 
	\ew{\kstar{k}}}$]
	\label{lem:star}
	Let $H$ be the $\kstar{k}$ subgraph, and let $\alpha>2$. If $\ew{H}$ can be 
	solved in $O(n^\alpha)$ time, then $\ksum{k}$ can be solved in 
	$\Ot(n^{(1-1/k)\cdot\alpha})$ time.
\end{lemma}

\begin{proof}
	To prove the lemma we define the problem $\ksumn$ to be the following. Given a 
	sequence of $n$ $\tableksum{k}$ instances, each on $n$ numbers, does there 
	exist an instance in the sequence that has a solution of sum 0?
	Then, we prove two claims, one showing a reduction from $\ksumn$ to 
	$\ew{\kstar{k}}$, and the other showing a self-reduction for $\ksum{k}$ that 
	relates it to $\ksumn$.

	\begin{claim}
	Let $H$ be the $\kstar{k}$ subgraph. $\ksumn$ can be reduced to the $\ew{H}$ 
	problem on a graph of $n$ nodes.
	\end{claim}

	\begin{proof}
		
		Given a $\ksumn$ instance, denote the $i^{th}$ $\ksum{k}$ instance in the 
		sequence as $L^{(i)}_1, \ldots, L^{(i)}_k$, where the $j^{th}$ list of the 
		$i^{th}$ instance is $L_j^{(i)} = \{ x^{(i)}_{j,\ell} \}_{\ell \in [n]} 
		\subseteq \mathbb{Z}$. We create an $\hpg{H}$ $G$ on $kn$ nodes, where we 
		associate the $\ith$ vertex in $\sn$ $\P_{h_{k+1}}$ (vertex $v_{k+1,i}$), 
		with the $\ith$ instance of the sequence, and we assign the $n$ numbers of 
		list $L^{(i)}_j$ to the $n$ edges incident to $v_{k+1,i}$ within the $\se$ 
		$(\P_{h_j}, \P_{h_{k+1}})$.	This can be done by setting, for every $i\in 
		[n], j\in [k], \ell \in [n]$, $w(v_{k+1,i} , v_{j,\ell}) = 
		x^{(i)}_{j,\ell}$, and the weight of every vertex in $G$ to $0$.
	
	Assume there is an $H$-subgraph $\chi = \hsub{j}{a}{k+1}$ in $G$ of total 
	weight $0$. Let $i=a_{k+1}$, and consider the $\ksol{k}$ for the $i^{th}$ 
	instance $\ksolsuper{i}{j}{a}{k}$. Note that its sum is exactly the total 
	weight of the $H$-subgraph $\chi$, which is $0$. For the other direction, 
	assume the $i^{th}$ instance has a $\ksol{k}$ $\ksolsuper{i}{j}{a}{k}$ of sum 
	$0$, and define the $H$-subgraph $\chi = \hsub{j}{a}{k+1}$, where $a_{k+1} = 
	i$. Again, note that the total weight of $\chi$ is exactly $\sum_{j=1}^k 
	x^{(i)}_{j,a_j} = 0$.
	\end{proof}

\begin{claim}
	Let $k\geq2$, and $\alpha\geq2$.
	If $\ksumn$ can be solved in $O(n^\alpha)$ time, then $\ksum{k}$ can be solved 
	in $\Ot(n^{\alpha \cdot \frac{k-1}{k}})$ time.
\end{claim}

\begin{proof}

We will use the hashing scheme due to Dietzfelbinger \cite{Die96} that we described and 
used in Appendix~\ref{app:conv}, to hash the numbers into buckets. Given a 
$\ksum{k}$ instance $L_1,\ldots,L_k$, our reduction is as follows:
\begin{enumerate}

\item Repeat the following $c \cdot k^k \cdot \log n$ times.
\begin{enumerate}

\item Pick a hash function $h \in \mathcal{H}_{M,t}$, for $t$, and map each 
	number $x_{i,j}$ to bucket $B_{i,h(x_{i,j})}$.

\item Ignore all numbers mapped to ``overloaded'' buckets.

\item Now each bucket has at most $N = kn/t=k\cdot n^{\frac{k-1}{k}}$ numbers. 
	We will generate a sequence of $k\cdot t^{k-1}=k\cdot n^{\frac{k-1}{k}}=N$ 
	instances of $\ksum{k}$, each on $N$ numbers, such that one of them has a 
	solution iff the original $\ksum{k}$ input has a solution. This sequence will 
	be the input to $\ksumN$:

Go over all $t^{k-1}$ choices of $k-1$ buckets, 
$B_{1,a_1},\ldots,B_{k-1,a_{k-1}}$, and add $k$ instances of $k$ to the 
sequence, one for each of the $k$ buckets, $B_{k,a^{(1)}},\ldots,B_{k,a^{(k)}}$, 
for which the last number in the solution might be in.
\end{enumerate}
\end{enumerate}

First note that if the $\ksum{k}$ had a solution, it's numbers will be mapped 
into not ``overloaded" buckets in one of the iterations, with probability 
$1-O(n^{-c})$. Then, in such case, the reduction will succeed due to the 
``almost linearity'' property of the hashing. Now, to conclude the proof of the 
claim, assume $\ksumN$ can be solved in time $O(N^\alpha)$, and observe that 
using the reduction one gets an algorithm for $\ksum{n}$ running in time 
$\Ot(N^\alpha) = \Ot(n^{\alpha\cdot\frac{k-1}{k}})$, as claimed. 

\end{proof}
\end{proof}

\medskip

Our final reduction between $\ksum{k}$ and $\ew{H}$ for a class of subgraphs $H$ 
is as follows. First, define the $\kpath{k}$ subgraph $H$ to be such that 
$V(H)=\{h_1,\ldots,h_{k}\}$ and 
$E(H)=\{(h_1,h_2),(h_2,h_3),\ldots,(h_{k-1},h_{k})\}$.

\begin{lemma}[${\ksum{k} \rto{\cftwo{k}}{\cftwo{k}} 
	\ew{\kpath{(k\text{-1})}}}$]
	\label{lem:path}
	Let $H$ be the $\kpath{k}$ subgraph. If $\ew{H}$ can be solved in time 
	$O(n^{\ctwo{k}-\varepsilon})$ for some $\varepsilon>0$, then $\ksum{k$+1$}$ 
	can be solved in time $O(n^{\ctwo{k} - \varepsilon'})$, for some 
	$\varepsilon'>0$.\end{lemma}

\begin{proof}
		We prove that an instance of $\convkpsum$ on $n$ numbers can be reduced to a 
		single instance of $\ew{\kpath{k}}$, and by applying 
		Lemma~\ref{lem:convksum}, this completes the proof.	  Given $k+1$ lists 
		$L_1,\ldots,L_{k+1}$ each with $n$ numbers as the input to $\convkpsum$, we 
		will construct a complete $\hpg{H}$ $G$ on $kn$ nodes. For every $r$ and $s$ 
		such that $r-s \in [n]$, for all $i \in [k]$, define the edge weights of $G$ 
		in the following manner.
	
	\[ w(v_{i,r}, v_{i+1,s}) =
		\begin{cases}
			x_{1,r} + x_{2,s-r}, &\text{if $i=1$} \\
			x_{i+1,s-r}, &\text{if $1 < i < k$} \\
			x_{k,s-r} + x_{k+1,s}, &\text{if $i=k$}
		\end{cases}
	\]

	Otherwise, if $r-s \not\in [n]$, we set $w(v_{i,r}, v_{i+1,s}) = -\infty$ for 
	all $i \in [k]$. Now to see the correctness of the reduction, take any 
	$H$-subgraph $\hsub{i}{a}{k}$ of $G$, and consider the $\ksol{(k+1)}$ 
	$\ksols{i}{b}{k+1}$, where $b_1=a_1$, $b_{k+1} = a_{k}$, and for $2 \leq i 
	\leq k$, $b_i=a_i-a_{i-1}$. First, note that the $\ksol{(k+1)}$ satisfies the 
	property that $b_1 + \ldots + b_{k} = b_{k+1}$. Now, note that its total 
	weight is $\sum_{i=1}^{k} w(\hsubv{i}{a}, \hsubv{i+1}{a}) = x_{1,a_1} + 
	x_{2,a_2-a_1} + x_{3,a_3 - a_2} + \ldots + x_{k-1,a_{k-1} - a_{k-2}} +  x_{k, 
	a_{k}-a_{k-1}} + x_{k+1,a_{k}} = \sum_{i=1}^{k+1} x_{i,b_i}$ which is exactly 
	the sum of the $\ksol{(k$+1$)}$. For the other direction, consider the 
	$\ksol{(k$+1$)}$ $\ksols{i}{b}{k+1}$ for which $b_{k+1} = b_1 
	+\ldots+b_{k+1}$. Then, the $H$-subgraph $\hsub{i}{a}{k}$, where $a_i = b_1 + 
	\ldots + b_i$, has total weight $\sum_{i=1}^{k+1} x_{i,b_i}$. Therefore, there 
	is a $\ksol{k}$ of sum $0$ iff there is an $H$-subgraph in $G$ of total weight 
	$0$. 
\end{proof}


\section{Relationships Between Subgraphs} \label{sec:vm}

In this section we prove Theorem~\ref{thm2} showing that 
$\ew{H_1}$ can be reduced to $\ew{H_2}$ if $H_1$ is a $\emph{\vertminor}$ of 
$H_2$. Then we give an additional 
observation that gives a reverse reduction. We start by defining $\vertminor$s.

\begin{definition}[Vertex-Minor]
\label{def:vm}
	A graph $H_1$ is called a $\vertminor$ of graph $H_2$, and denoted $H_1 \vm 
	H_2$ , if there exists a sequence of subgraphs $H^{(1)}, \ldots, H^{(\ell)}$ 
	such that $H_1 = H^{(1)}$, $H_2 = H^{(\ell)}$, and for every $i \in [\ell-1]$, 
	$H^{(i)}$ can be obtained from $H^{(i+1)}$ by either
	\begin{itemize}
		\item Deleting a single edge $e \in E(H^{(i+1)})$, or
		\item Contracting two nodes\footnote{The difference between our definition 
			of $\vertminor$ and the usual definition of a graph minor is that we allow 
			contracting two nodes that are not necessarily connected by an edge.} 
			$h_j,h_k \in V(H^{(i+1)})$ to one node $h_{jk} \in V(H^{(i)})$, such that 
			$N(h_{jk}) = N(h_j) \cup N(h_k)$.
	\end{itemize}
\end{definition}

To prove Theorem~\ref{thm2} it suffices to show how to reduce $\ew{H_1}$ to 
$\ew{H_2}$ when $H_1$ is obtained by either a single edge deletion or a single 
vertex contraction.
The edge deletion reduction is straightforward, and the major part of the proof will be 
showing the contraction reduction.
The main observation is that we can make two copies of nodes and change their 
node weights in a way such that any $H_2$-subgraph of total weight 0 that 
contains one of the copies will have to contain the other. This will allow us to 
claim that the subgraph obtained by replacing the two copies of a node with the 
original will be an $H_1$-subgraph of total weight 0.

\begin{lemma}
	\label{lemma:edgedel}
	Let $H$ be a subgraph you get after deleting an edge from $H'$. $\ew{H}$ on 
	$kn$ nodes can be reduced to a single instance of $\ew{H'}$ on $kn$ nodes.
\end{lemma}

\begin{proof}
	Without loss of generality, denote $V(H)=V(H')=\{h_1,\ldots, h_k\}, 
	E(H')=\{e_1,\ldots,e_m\}$, where $e_m=(h_{k-1},h_k)$, and 
	$E(H)=\{e_1,\ldots,e_{m-1}\}$.
	
	Given $G$, an $\hpg{H}$ as input to $\ew{H}$, we create an $\hpg{H'}$ $G'$ 
	which will have the same set of nodes as $G$, but will have an additional 
	super-edge $(h_{k-1},h_k)$ where all of the edges within this super-edge will 
	have weight 0. In other words, for all $a,b\in [n]$ define 
	$w(v_{(k-1),a},v_{k,b})=0$. Now, every $H$-subgraph in $G$ is an $H'$-subgraph 
	in $G'$ with the same total weight, and vice versa, which proves the 
	correctness of the reduction.
\end{proof}

\begin{lemma}
	\label{lemma:contraction}
	Let $H$ be a subgraph you get after contracting two nodes from $H'$. $\ew{H}$ 
	on $(k+1)n$ nodes can be reduced to a single instance of $\ew{H'}$ on $kn$ 
	nodes.
\end{lemma}

\begin{proof}
	Without loss of generality, denote $V(H')=\{h_1,\ldots, h_{k-1} \} \cup \{ 
	h_{\kone}, h_{\ktwo}\}, V(H)=\{ h_1,\ldots,h_{k-1}\} \cup \{ h_k \}$, and 
	assume you get $H$ from $H'$ by contracting the nodes $h_{\kone}, h_{\ktwo} 
	\in V(H')$ into the node $h_k \in V(H)$. Given $G$, an $\hpg{H}$, we create an 
	$\hpg{H'}$ $G'$ which will be almost the same as $G$, except that every vertex 
	$v_{k,a}$ in the $\kth$ partition of $G$ will have two copies in $G'$, one in 
	$\P_{h_{\kone}}$ and one in $\P_{h_{\ktwo}}$, which we call $v_{\kone,a}$ and 
	$v_{\ktwo,a}$, respectively. The weights in $G'$ will be the same as in $G$, 
	but we will add a unique integer $u_a$ for $a \in [n]$ to the weight of 
	$v_{\kone,a}$ and subtract $u_a$ from the weight of $v_{\ktwo,a}$. This will 
	ensure that in any $H'$-subgraph of total weight $0$, if $v_{\kone,a}$ is 
	picked, then $v_{\ktwo,a}$ must also be picked. This allows us to conclude 
	that any $H'$-subgraph of total weight $0$ in $G'$ will directly correspond to 
	an $H$-subgraph in $G$.
	
	Let $d=|E(H)|=|E(H')|$, $W$ the maximum weight of any edge or node in $G$,	
	and $K = (d+k+1) \cdot  W$. Create a complete $\hpg{H'}$ $G'$, and define the 
	edge weights $w':(E(G')\cup V(G')) \rightarrow \mathbb{Z}$ as follows. For 
	every $\se$ $(h_i, h_{\khat})$ where $\khat \in \{ \kone, \ktwo\}$, define 
	$w'(v_{i,a}, v_{\khat,b}) = w(v_{i,a}, v_{k,b})$. All other edges $(v_{i,a}, 
	v_{j,b}) \in E(H')$ will have weight $w'(v_{i,a}, v_{j,b}) = w(v_{i,a}, 
	v_{j,b})$. For the vertices, we will set $w'(v_{\kone,a}) = a \cdot K$ and 
	$w'(v_{\ktwo,a}) = -a \cdot K$ for all $a \in [n]$. All other vertices will 
	have weight $0$.

	Let $\chi = \hsub{i}{a}{k-1} \cup \{ v_{{\kone},a_{\kone}}, 
	v_{{\ktwo},a_{\ktwo}} \}$ be an $H'$-subgraph of $G'$ of total weight $0$. 
	First, we claim that $a_{\kone} = a_{\ktwo}$. This is true because the total 
	weight of the subgraph $\chi$ is $(a_{\kone} - a_{\ktwo})\cdot K + X$, where 
	$X$ represents the sum of $d$ edges and $k$ nodes, each of weight at most $W$. 
	Therefore, $X < (d+k+1) \cdot W = K$, which implies that $(a_{\kone} - 
	a_{\ktwo}) \cdot K + X = 0$ can happen only if $a_{\kone} - a_{\ktwo}=0$. 
	Second, note that the $H$-subgraph $\hsub{i}{a}{k}$ of $G$, where 
	$a_k=a_{\kone}=a_{\ktwo}$, will also have total weight $0$. This is because 
	the numbers added to the weights of the nodes $v_{k^{(1)}, a_{k^{(1)}}}$ and 
	$v_{k^{(1)}, a_{k^{(1)}}}$ cancel out, and all of the other weights involved 
	are defined to be the same as in $G$. Now for the other direction, note that 
	for any $H$-subgraph $\hsub{i}{a}{k}$ in $G$, the $H'$-subgraph $\{ v_{i,a_i} 
	\}_{i\in [k-1]} \cup \{ v_{{\kone},a_{\kone}}, v_{{\ktwo},a_{\ktwo}} \}$ in 
	$G'$, where $a_{\kone} = a_{\ktwo} = a_k$, will have the same total weight. 
	Therefore, there is an $H'$-subgraph of total weight $0$ in $G'$ if and only 
	if there is an $H$-subgraph of total weight $0$ in $G$. 
\end{proof}

Theorem~\ref{thm2} follows from these two lemmas, by the 
transitive property of our reducibility definition.

\subsection{Reverse Direction}

Next we give an observation which shows how $\ew{H_2}$ can be reduced to 
$\ew{H_1}$, where $H_2$ contains $H_1$ as an induced subgraph. This can be seen 
as a reversal of Theorem~\ref{thm2}, since $H_2$ is a larger graph.

\begin{proposition}
	\label{lemma:addingvertices}
	Let $H_2$ be the subgraph you get by adding a node to $H_1$ that has edges to 
	every node in $H_1$, and let $\alpha \geq 2$. Then, $\ew{H_2} 
	\rto{\alpha+1}{\alpha} \ew{H_1}$ .
\end{proposition}

\begin{proof}
	Without loss of generality denote $V(H')=\{h_1,\ldots,h_k\}$, 
	$V(H)=\{h_1,\ldots,h_k,h_{k+1}\}$ and $E(H)=E(H')\cup 
	\{(h_i,h_{k+1})\}_{i=1}^k$. That is, $h_{k+1}$ is the new node, and it's 
	connected with edges to all the other nodes.
	
	To solve $\ew{H}$ on a complete $\hpg{H}$ $G$ on $(k+1)\cdot n$ nodes, we will 
	create $n$ instances of $\ew{H'}$, one for every vertex $v_{k+1,a}$ in the 
	$\sn$ $\P_{h_{k+1}}$. The instance will have a solution if and only if there 
	is an $H$-subgraph in $G$ of total weight $0$ that has node $v_{k+1,a}$ in it:
	
	For every $a\in [n]$, create an $\hpg{H'}$ $G'_a$ that will be the same as $G$ 
	on all $\sn$s and all $\se$s that do not involve $h_{k+1}$, but will have the 
	following additional weights:
	For every $i \in [k]$ and $b \in [n]$, we will add the weight of the edge 
	$(v_{i,b},v_{k+1,a})$ in $G$, to the weight of node $v_{i,b}$ in $G'$. 
	
	Now observe that every $H'$-subgraph in $G'_a$, $\hsub{i}{a}{k}$ will have 
	exactly the same weight as the $H$-subgraph of $G$ which is 
	$\hsub{i}{a}{k+1}$, where $a_{k+1}=a$. And therefore, there is an 
	$H'$-subgraph in $G'_a$ of total weight $0$, if and only if there is an 
	$H$-subgraph of $G$ which has the vertex $v_{k+1,a}$ and has weight 0.
\end{proof}


\vspace{-0.1in}

\section{Reductions to $\ksum{k}$ (and Upper Bounds)} \label{sec:algs}

In this section we show how $\ksum{k}$ can be used to solve $\ews$ subgraph 
problems. First, we show how the standard reduction from clique detection to 
$\ksum{k}$ can be generalized to relate $\ksum{k}$ to $\ews$ subgraph problems. 
However, this reduction does not give non-trivial implications for most 
subgraphs. Then, we show how to use a $\ksum{2}$ algorithm to solve $\ew{H}$ for 
any subgraph $H$. This gives us a generic algorithm for solving $\ew{H}$, which 
we call the $\sepalg$. Finally, we generalize this algorithm in a way that 
allows it to be phrased as a reduction from $\ew{H}$ to $\ksum{d}$ for any $d 
\leq k-1$, where $k = |V(H)|$. The bounds achieved by the algorithms depend on 
the structure of $H$.

In \cite{DF95}, a reduction that maps an unweighted $\kclique{k}$ detection 
instance to a $\ksum{k \choose 2}$ instance on $n^2$ numbers is given in order 
to prove that $\ksum{k}$ is $W[1]$-hard. The reduction maps each edge to a 
number that encodes the two vertices that are adjacent to the edge in a way such 
that the numbers encoded in the edges corresponding to a $\kclique{k}$, when 
summed, cancel out to $0$. In \cite{JV13}, the authors show how the same idea 
can be applied to show that triangle detection can be reduced to $\kxor{3}$ on 
$n^2$ vectors. We show that $\ew{H}$ can be reduced to $\ksum{d}$ on $n^2$ 
numbers, where $d = |E(H)|$.
 
\begin{proposition}
\label{prop:folk}
Let $H$ be a subgraph on $k$ nodes and $d$ edges, and let $\alpha \geq 2$. Then, 
$\ew{H} \rto{2\cdot\alpha}{\alpha} \ksum{d}$.
Moreover, the $\ew{H}$ problem on graphs with $m$ edges can be reduced to 
$\ksum{d}$ on $m$ numbers. \end{proposition}

\begin{proof}
We give a simple proof that uses out techniques.
First note that any $H$ with $d$ edges is a $\vertminor$ of the $\kmatching{d}$ 
subgraph, and therefore by Theorem~\ref{thm2},  $\ew{H} \rto{\alpha}{\alpha} 
\ew{\kmatching{d}}$, and note that in our reduction, the sparsity of the graph 
is preserved. Then use Theorem~\ref{thm3} to reduce $\ew{\kmatching{d}}$ to 
$\ksum{d}$, by choosing $S=\emptyset$ and $H_i$ to be the two endpoints of the 
$\ith$ edge, and observe that we get a $\ksum{d}$ instance on $m$ numbers.
\end{proof}

\paragraph{The Separator Algorithm.} We will say that $(S, H_1, H_2)$ is a 
\emph{separator} of a graph $H$ iff $S, H_1, H_2$ partition $V(H)$ and there are 
no edges between a vertex in $H_1$ and a vertex in $H_2$. The set of all 
separators of $H$ will be denoted as $\sepset(H)$. Consider the following 
algorithm to solve $\ew{H}$. We will call this algorithm the \emph{separator 
algorithm}. First, find \[ (S, H_1, H_2) = \argmin_{(S', H_1', H_2') \in 
\sepset(H)} (|S'| + \max(|H_1'|, |H_2'|)) \] by naively brute-forcing over all 
$3^k$ possible choices for $S$, $H_1$, and $H_2$. Then, pick an $S$-subgraph 
$\chi_S = \hsubsia{S}$. Construct a $\ksum{2}$ instance with target weight 
$w(\chi_S)$ and lists $L_1$ and $L_2$ constructed as follows: For every 
$H_1$-subgraph $\chi_{H_1} = \hsubsia{H_1}$, add $w(\chi_{H_1} \cup \chi_S)$ to 
$L_1$. Similarly, for every $H_2$-subgraph $\chi_{H_2} = \hsubsia{H_2}$, add 
$w(\chi_{H_2} \cup \chi_S)$ to $L_2$. We create an instance of $\ksum{2}$ for 
each possible $S$-subgraph $\chi_S$. The algorithm outputs that there is an 
$H$-subgraph of total weight $0$ iff some $\ksum{2}$ instance has a solution. 
The running time of this algorithm is $O(3^k + n^{|S|} \cdot (n^{|H_1|} + 
n^{|H_2|}))$. We will call $\gamma(H) = |S| + \max(|H_1|, |H_2|)$.

	To see that the separator algorithm solves $\ew{H}$ in time 
	$O(n^{\gamma(H)})$,
suppose there is an $H$-subgraph $\hsub{i}{a}{k} = \hsubsia{S} \cup 
\hsubsia{H_1} \cup \hsubsia{H_2}$ of weight $0$. Then, the $\ksum{2}$ instance 
corresponding to $\chi_S = \hsubsia{S}$ will contain an integer $w(\hsubsia{H_1} 
\cup \hsubsia{S}) \in L_1$ and $w(\hsubsia{H_2} \cup \hsubsia{S}) \in L_2$.

Since $S, H_1, H_2$ partition $V(H)$, we can write the sum of these two integers 
as \[ w(\hsubsia{S} \cup \hsubsia{H_1} \cup \hsubsia{H_2}) + w( \hsubsia{S}). \] 
Note that the term on the left is equal to $w(\hsubsia{V(H)})$, which by our 
assumption is $0$. Thus, the sum of these two integers is equal to the target, 
$w( \hsubsia{S})$.

For the other direction, let $\chi_S = \hsubsia{S}$ be the corresponding 
$S$-subgraph to the $\ksum{2}$ instance which has a solution of the form $w( 
\hsubsia{H_1} \cup \hsubsia{S}) \in L_1$ and $w(\hsubsia{H_2} \cup 
\hsubsia{S})\in L_2$. Since $S,H_1,H_2$ partition $V(H)$, we can write the sum 
of these two integers as:
\[ w( \hsubsia{H_1} \cup \hsubsia{S} \cup \hsubsia{H_2}) + w(\hsubsia{S}). \]
Since the target of the $\ksum{2}$ instance is $w( \hsubsia{S})$, we have that 
$w(\hsubsia{H_1} \cup \hsubsia{S} \cup \hsubsia{H_2}) = w(\hsubsia{V(H)}) = 0$. 

\begin{remark}
	The $\sepalg$ is quite simple, yet we are not aware of any subgraph $H$ for 
	which there is an algorithm that solves $\ew{H}$ in time $O(n^{\gamma(H) - 
	\varepsilon})$, for some $\varepsilon>0$. We have given examples of subgraphs 
	for which improving on the $\sepalg$ is known to imply that the $\ksumconj{k}$ 
	is false, and some for which this implication is not known.
\end{remark}

\paragraph{Generalizing the Separator Algorithm.} We can view the separator 
algorithm as an algorithm which finds the optimal way to ``break'' $H$ into two 
subgraphs $H_1$ and $H_2$, and enumerates all instances of $H_1$ and $H_2$ 
independently, and then solves $2$-sum instances to combine the edge-disjoint 
subgraphs. One natural way to generalize this algorithm is to consider what 
happens when we divide $H$ into $d$ subgraphs $H_1, \ldots, H_d$. Then, by a 
similar algorithm, one can use $\ksum{d}$ to solve the $\ew{H}$ problem. This 
generalization is of interest due to the fact that it implies that faster 
$\ksum{d}$ algorithms imply faster algorithms for $\ew{H}$.

We will say that $(S, H_1, \ldots, H_d)$ is a \emph{$d$-separator} iff $S, H_1, 
\ldots, H_d$ partition $V(H)$ and there are no edges between a vertex in $H_i$ 
and a vertex in $H_j$ for any distinct $i,j \in [1,d]$. The set of all 
$d$-separators of $H$ will be denoted as $\sepset^d(H)$.


\paragraph{Reminder of Theorem~\ref{thm3}:} Let $(S, H_1, \ldots, H_d)$ be a 
$d$-separator of $H$. Then, $\eweight{H}$ can be reduced to $\Ot(n^{|S|})$ 
instances of $\ksum{d}$ each on $\max\{n^{|H_1|}, \ldots, n^{|H_d|} \}$ numbers.
\begin{proof}[of Theorem~\ref{thm3}]
We can generalize the separator algorithm to hold for arbitrary $d$-separators. 
Pick an $S$-subgraph $\chi_S = \hsubsia{S}$. Construct a $\ksum{d}$ instance 
with target weight $(d-1) \cdot w(\chi_S)$ and lists $L_1, \ldots, L_d$ 
constructed as follows: For all $j \in [d]$, for every $H_j$-subgraph 
$\chi_{H_j} = \hsubsia{H_j}$, add $w(\chi_{H_j} \cup \chi_S)$ to $L_j$. We 
create an instance of $\ksum{d}$ for each possible $S$-subgraph $\chi_S$. The 
algorithm outputs that there is an $H$-subgraph of total weight $0$ iff some 
$\ksum{d}$ instance has a solution. The proof of correctness for this reduction 
follows similarly to the proof of correctness for the separator algorithm. The 
$\Ot(\cdot)$ in the number of instances comes from the $\eweight{H} \rto{}{} 
\ew{H}$ reduction.

\end{proof}

\begin{corollary}
	\label{cor:IS}
	Let $H$ be a graph on $k$ nodes and let $I$ be an independent set of $H$ where 
	$|I| = s$. Then, $\ew{H}$ can be reduced to $O(n ^ { k - s})$ instances of 
	$\ksum{s}$ on $n$ integers.
\end{corollary}
\begin{proof}
	Consider the separator $(S,H_1,\ldots,H_s)$, where $S=V(H)\setminus I$, and 
	$H_i$ is a singleton containing the $i^{\text{th}}$ vertex in $I$.
	
\end{proof}

\begin{corollary}
	\label{cor:TW}
	Let $H$ be a graph on $k$ nodes with treewidth bounded by $tw$. Then, $\ew{H}$ 
	can be solved in time $O(n^{\frac{2}{3}k+tw})$.
\end{corollary}
\begin{proof}
Observe that there will be a $d$-separator of size $tw$, for some $d>1$, where 
each of the $d$ disconnected components has at most $k/2$ nodes, and therefore 
can be separated into two disconnected components with at most $2k/3$ nodes 
each.	Thus, $\gamma(H) \leq 2k/3 + tw$.

\end{proof}

\begin{remark}
\label{rem:sepalgopt}
Note that under the current best known running times for $\ksum{k}$, the 
$\sepalg$ (of Corollary~\ref{cor4}) will always be at least as good as the 
algorithm one can get from Theorem~\ref{thm3}. This is implied by the fact that 
the fastest known way to solve $\ksum{k}$ is by a reduction to $\ksum{2}$. 
However, if it turns out that there exists a $k_0$ for which $\ksum{k_0}$ can be 
solved fast enough, the algorithm of Theorem~\ref{thm3} can be faster than the 
$\sepalg$. As an example, assume $\ksum{3}$ can be solved in linear time, and 
$H$ is a subgraph composed of $3$ disconnected $k/3$-node cliques.
\end{remark}

Now notice that if one wanted to find the minimum total weight of an 
$H$-subgraph in the input graph $G$, the same procedure can be applied, with a 
slight modification that makes it more efficient.
When going over an $S$-subgraphs $\chi_S$ of $G$, instead of solving $\ksum{d}$ 
on the $d$ lists $L_1,\ldots,L_d$, it is enough to find the minimum number in 
each list. Observe that the sum of these numbers, minus $(d-1)\cdot w(\chi_S)$, 
equals the minimum total weight of an $H$-subgraph in $G$ that uses the nodes in 
$\chi_S$.  Therefore, by going over all $S$-subgraphs, and taking the minimum of 
these numbers, one gets the minimum total weight of an $H$-subgraph in $G$. The 
running time of this modified procedure is 
$O(n^{|S|}\cdot(n^{|H_1|}+\cdots+n^{|H_d|}))$.

	\paragraph{Reminder of Theorem~\ref{thm4}:}
	Let $H$ be a subgraph on $k$ nodes, with independent set of size $s$. Given a 
	graph $G$ on $n$ nodes with node and edge weights, the minimum total weight of 
	a (not necessarily induced) subgraph of $G$ that is isomorphic to $H$ can be 
	found in time $\Ot(n^{k-s+1})$.

\begin{proof}
	First, observe that by our proof of the reduction $\eweight{H} \rto{}{} 
	\ew{H}$ in Appendix~\ref{app:hpg}, an algorithm for the minimization problem 
	that assumes the graph is $H$-partite yields an algorithm for the original 
	problem with the same running time, up to $k^k\cdot \poly\log n$ factors.
	Then, use the procedure mentioned above where $(S,H_1,\ldots,H_d)$ are as in 
	the proof of Corollary~\ref{cor:IS}, to solve the structured version of the 
	problem in time $O(n^{k-s}\cdot n)$.

\end{proof}


\vspace{-0.1in}

\section{Conclusions}

We conclude with two interesting open questions:
\begin{enumerate}
\item Perhaps the simplest subgraph for which we cannot give tight lower and 
	upper bounds is the $\kcycle{5}$ subgraph. Can we achieve $O(n^{4-\varepsilon})$ for some $\varepsilon>0$ without 
	breaking the $\ksumconj{k}$, or can we prove that it is not possible?
\item Can we prove that $\eweight{\kpath{4}} \rto{3}{3} \eweight{\kstar{3}}$? 
	This would show that breaking the $\ksumconj{3}$ will imply an $O(n^{3-\varepsilon})$ for some $\varepsilon>0$ 
	algorithm for $\apsp$.
\end{enumerate}

\paragraph{Acknowledgements.}
The authors would like to thank Ryan and Virginia Williams for many helpful 
discussions and for sharing their insights, and Hart Montgomery for initiating 
the conversation that led up to this work. We would also like to thank the 
anonymous reviewers for their comments and suggestions.

\bibliographystyle{plain}
\normalsize \bibliography{eweightfullversion}

\begin{thebibliography}{10}

\bibitem{AC05}
Nir Ailon and Bernard Chazelle.
\newblock Lower bounds for linear degeneracy testing.
\newblock {\em J. ACM}, 52(2):157--171, 2005.

\bibitem{AYZ95}
Noga Alon, Raphael Yuster, and Uri Zwick.
\newblock Color-coding.
\newblock {\em J. ACM}, 42(4):844--856, July 1995.

\bibitem{BDP08}
Ilya Baran, Erik~D. Demaine, and Mihai P{\v a}tra{\c s}cu.
\newblock Subquadratic algorithms for {3SUM}.
\newblock {\em Algorithmica}, 50(4):584--596, 2008.
\newblock See also WADS'05.

\bibitem{Die96}
Martin Dietzfelbinger.
\newblock Universal hashing and k-wise independent random variables via integer
  arithmetic without primes.
\newblock In Claude Puech and R{\"u}diger Reischuk, editors, {\em STACS},
  volume 1046 of {\em Lecture Notes in Computer Science}, pages 569--580.
  Springer, 1996.

\bibitem{DF95}
Rod~G. Downey and Michael~R. Fellows.
\newblock Fixed-parameter tractability and completeness ii: On completeness for
  w[1], 1995.

\bibitem{EG04}
Friedrich Eisenbrand and Fabrizio Grandoni.
\newblock On the complexity of fixed parameter clique and dominating set.
\newblock {\em Theor. Comput. Sci.}, 326(1-3):57--67, 2004.

\bibitem{Eri95}
Jeff Erickson.
\newblock Lower bounds for linear satisfiability problems.
\newblock In Kenneth~L. Clarkson, editor, {\em SODA}, pages 388--395. ACM/SIAM,
  1995.

\bibitem{FLRSR12}
Fedor~V. Fomin, Daniel Lokshtanov, Venkatesh Raman, Saket Saurabh, and
  B.~V.~Raghavendra Rao.
\newblock Faster algorithms for finding and counting subgraphs.
\newblock {\em J. Comput. Syst. Sci.}, 78(3):698--706, 2012.

\bibitem{GO95}
Anka Gajentaan and Mark~H Overmars.
\newblock On a class of $o(n^2)$ problems in computational geometry.
\newblock {\em Computational Geometry}, 5(3):165 -- 185, 1995.

\bibitem{IR77}
Alon Itai and Michael Rodeh.
\newblock Finding a minimum circuit in a graph.
\newblock In {\em STOC}, STOC '77, pages 1--10, New York, NY, USA, 1977. ACM.

\bibitem{JV13}
Zahra Jafargholi and Emanuele Viola.
\newblock 3sum, 3xor, triangles.
\newblock {\em Electronic Colloquium on Computational Complexity (ECCC)}, 20:9,
  2013.

\bibitem{KKM00}
Ton Kloks, Dieter Kratsch, and Haiko M{\"u}ller.
\newblock Finding and counting small induced subgraphs efficiently.
\newblock {\em Inf. Process. Lett.}, 74(3-4):115--121, 2000.

\bibitem{KLL11}
Miroslaw Kowaluk, Andrzej Lingas, and Eva-Marta Lundell.
\newblock Counting and detecting small subgraphs via equations and matrix
  multiplication.
\newblock In {\em SODA}, SODA '11, pages 1468--1476. SIAM, 2011.

\bibitem{NP85}
Jaroslav Nešetřil and Svatopluk Poljak.
\newblock On the complexity of the subgraph problem.
\newblock {\em Commentationes Mathematicae Universitatis Carolinae},
  026(2):415--419, 1985.

\bibitem{Pat10}
Mihai P{\v a}tra{\c s}cu.
\newblock Towards polynomial lower bounds for dynamic problems.
\newblock In {\em Proc. 42nd ACM Symposium on Theory of Computing (STOC)},
  pages 603--610, 2010.

\bibitem{PW10}
Mihai P{\v a}tra{\c s}cu and Ryan Williams.
\newblock On the possibility of faster sat algorithms.
\newblock In {\em Proc. 21st ACM/SIAM Symposium on Discrete Algorithms (SODA)},
  pages 1065--1075, 2010.

\bibitem{VW09}
Virginia Vassilevska and Ryan Williams.
\newblock Finding, minimizing, and counting weighted subgraphs.
\newblock In Michael Mitzenmacher, editor, {\em STOC}, pages 455--464. ACM,
  2009.

\bibitem{Wil09}
Ryan Williams.
\newblock Finding paths of length k in o$^{\mbox{*}}$(2$^{\mbox{k}}$) time.
\newblock {\em Inf. Process. Lett.}, 109(6):315--318, 2009.

\bibitem{WW10}
Virginia~Vassilevska Williams and Ryan Williams.
\newblock Subcubic equivalences between path, matrix and triangle problems.
\newblock In {\em FOCS}, pages 645--654. IEEE Computer Society, 2010.

\end{thebibliography}

\appendix

\section{Reducibility}\label{app:reducibility}

Our definition of reducibility is a mild extension of the definition of sub 
cubic reducibility in \cite{WW10}(Definition C.1). In weighted graph problems 
where the weights are integers in $[-M,M]$, $n$ will refer to the number of 
nodes times $\log M$. For $\ksum{k}$ problems where the input integers are in 
$[-M,M]$, $n$ will refer to the number of integers times $\log M$.

\begin{definition}
	Let $A$ and $B$ be two decision problems. We say that $A \rto{a}{b} B$, if 
	there is an algorithm $\mathcal{A}$ with oracle access to $B$, such that for 
	every $\varepsilon>0$ there is a $\delta>0$ satisfying three properties:
	\begin{itemize}
	\item For every instance $x$ of $A$, $\mathcal{A}$ solves the problem $A$ on 
		$x$  probability $1-o(1)$.
	\item $\mathcal{A}$ runs in time $O(n^{a-\delta})$ time on instances of size 
		$n$.
	\item For every instance $x$ of $A$ of size $n$, let $n_i$ be the size of the 
		$\ith$ oracle access to $B$ in $\mathcal{A}(x)$. Then $\sum_i 
		n_i^{b-\varepsilon} \leq n^{a-\delta}$.
	\end{itemize}
	
\end{definition}

The proofs of Propositions 1 and 2 in \cite{WW10}, prove that this definition 
has the following two properties that we will use:
\begin{itemize}
\item Let $A,B,C$ be problems so that $A \rto{a}{b} B$ and $B \rto{b}{c} C$, 
	then $A \rto{a}{c} C$.
\item If $A \rto{a}{b} B$ then an $O(n^{b-\varepsilon})$ algorithm for $B$ for 
	some $\varepsilon>0$, implies an $O(n^{a-\delta})$ algorithm for $A$ for some 
	$\delta>0$, that succeeds with probability $1-o(1)$.
\end{itemize}


\section{Proof of Lemma~\ref{lem:ewhequiv}} \label{app:hpg}

\begin{claim}
	$\ew{H} \rto{\alpha}{\alpha} \eweight{H}$.
\end{claim}

\def\zth{z^{th}}

\begin{proof}
	Let $G = \hsub{i}{a}{k}$ be an $H$-partite graph with weight function $w : 
	V(G) \cup E(G) \to \mathbb{Z}$ and target weight $0$. In this proof, we will 
	construct a new weight function $w^* : E(G) \to \mathbb{Z}$ and a target $t 
	\in \mathbb{Z}$ such that $G$, $H$, $w^*$, and $t$ make up an instance of 
	$\eweight{H}$. We will build $w^*$ in the following manner. Let $W$ be the 
	maximum weight of a node or edge in the graph, and let $d= W \cdot |E(H)|$. 
	First, initialize $w^*(v_{i,a}, v_{j,b}) = w(v_{i,a}, v_{j,b})$ for all edges 
	in $G$. Then, let $(P_{h_a}, P_{h_b})$ be the $\zth$ $\se$ of $H$. For each 
	edge $(v_{i,a}, v_{j,b})$, we add the integer $d^z$ to $w^*(v_{i,a}, 
	v_{j,b})$. Now, for each $i \in [n]$, pick an arbitrary $j \in [n]$ such that 
	$(h_i,h_j) \in E(H)$, and add $w(v_{i,a})$ to $w^*(v_{i,a},v_{j,b})$ for all 
	$a,b \in [n]$. We set the target $t = \sum_{i=1}^{d} d^i$.

	To prove correctness, let $\chi = \hsub{i}{a}{k}$ be an $H$-subgraph of $G$ of 
	total weight $0$. Since each $\se$ of $G$ is used exactly once by $\chi$, it 
	follows that the sum of the edges of $\chi$ will have total weight $t$ under 
	weight function $w^*$. For the reverse direction, let $S = \{ v_i \}_{i \in 
	[k]}$ be a subgraph isomorphic to $H$ whose edge weights sum to $t$ under 
	weight function $w^*$. Then, each $v_i$ must lie in a distinct $\sn$ of $G$, 
	for otherwise, if the $\jth$ $\sn$ is unoccupied, then the total weight of $S$ 
	cannot possibly sum to $t$. Now, relabel the vertices of $S$ as 
	$\hsub{i}{a}{k}$. Then, the total weight of $S$ under $w^*$ can be expressed 
	as $\sum_{h_i \in V(H)} w(v_{i,a_i}) + \sum_{(h_i,h_j) \in E(H)} w(v_{i,a_i}, 
	v_{j,a_j}) + t$. Therefore, we conclude that the sum of the weights of the 
	nodes and edges of $S$ is $0$, as desired.
\end{proof}

\begin{claim}
	$\eweight{H} \rto{\alpha}{\alpha} \ew{H}$.
\end{claim}

\begin{proof}
	We will use a simple color coding trick to ensure that the reduction succeeds 
	with probability $1/k^k$. This procedure can be derandomized using standard 
	techniques.

	Let $G^*$ be the graph of an instance of $\eweight{H}$. We construct an 
	$H$-partite graph $G$ in the following manner. For each vertex $v_i \in V(G)$, 
	we will pick a random $j \in [k]$ and put $v_i$ in $\sn$ $\P_{h_j}$. In other 
	words, we maintain the structure of the graph while partitioning the vertices 
	into $k$ parts. This can also be seen analogously as color-coding the vertices 
	using $k$ colors. The graph $G$ (along with the original weight function) is 
	now an instance of $\ew{H}$.

	For correctness, note that if $G$ contains an $H$-subgraph $\chi$, then $\chi$ 
	is an isomorphic copy of $H$ in $G^*$ with probability $1$. For the other 
	direction, we will show that with probability at least $1/k^k$, a set of 
	vertices $\chi = \{v_i \}_{i \in [k]}$ from $G^*$ will form an $H$-subgraph in 
	$G$. For each vertex $v_i$, there is a $1/k$ probability that it is assigned 
	to partition $\P_{h_i}$. Thus, with probability $1/k^k$, this event holds for 
	all $v_i$ for $i \in [k]$, and so $\chi$ is an $H$-subgraph of $G$. To 
	translate this into a reduction, we simply repeat this randomized procedure 
	$O(k^k)$ times.
\end{proof}


\section{Proof of Lemma~\ref{lem:convksum}}
\label{app:conv}

\begin{proof}

Assume $\convksum{k}$ can be solved in time $O(n^{\cftwo{k}-\varepsilon})$, for 
some $\varepsilon >0$. We follow the outline of the proof of Theorem 10 in 
\cite{Pat10} to give an $O(n^{\cftwo{k}-\varepsilon'})$ time algorithm for 
$\ksum{k}$.

We use a hashing scheme due to Dietzfelbinger \cite{Die96} to hash the numbers 
of the $\ksum{n}$ instance to $t$ buckets. In \cite{Die96}, a simple hash family 
$\mathcal{H}_{M,t}$ is given, such that if one picks a function 
$h:[M]\rightarrow [t]$ at random from $\mathcal{H}_{M,t}$, and maps each number 
$x_{i.j}\in L_i$ to bucket $B_{i,h(x_{i,j})}$, the following will hold:

\begin{itemize}
\item(Good load balancing) W.h.p. only $O(kt)$ numbers will be mapped to 
	``overloaded'' buckets, that is, buckets with more than $kn/t$ numbers. 
	Moreover, each number will be hashed to an ``overloaded" bucket with $o(1)$ 
	probability.
\item(Almost linearity) For any $k-1$ buckets 
	$B_{1,a_1},\ldots,B_{k-1,a_{k-1}}$, and any $k-1$ numbers $y_1\in 
	B_{1,a_1},\ldots,y_{k-1}\in B_{k-1,a_{k-1}}$, the number 
	$z=-(y_1+\cdots+y_{k-1})$ can only be mapped to one of certain $k$ buckets 
	(w.p. $1$): $B_{k,a^{(1)}},\ldots,B_{k,a^{(k)}}$, where w.l.o.g. we can assume 
	that $a^{(1)}=\sum_{i=1}^{k-1} a_{i}$, and for $1<i\leq k$, 
	$a^{(i)}=a^{(i-1)}+1$. \end{itemize}

Given a $\ksum{k}$ instance $L_1,\ldots,L_k$, our reduction is as follows:
\begin{enumerate}
\item Repeat the following $c \cdot k^k \cdot \log n$ times.
\begin{enumerate}

\item Pick a hash function $h \in \mathcal{H}_{M,t}$, for $t$ to be set later, and map each 
	number $x_{i,j}$ to bucket $B_{i,h(x_{i,j})}$.

\item Ignore all numbers mapped to ``overloaded'' buckets.
\item Now each bucket has at most $R = kn/t$ numbers. We create $k\cdot R^k$ 
	instances of $\convksum{k}$, one for every choice of numbers 
	$(i_1,\ldots,i_k)\in [R]^k$ and a number $0\leq y < k$, where in each 
	instance, the lists will contain only $t$ numbers. These instances will test 
	all $\ksol{k}$s that might lead to a solution.
For a fixed $(i_1,\ldots,i_k)\in [R]^k$ and $0\leq y < k$, we create $k$ lists 
$L'_1,\ldots,L'_k$ as input for $\convksum{k}$, where for every $j\in [k-1]$, 
$x'_{j,a} \in L'_j$ will be set to the $i_j$-th number of bucket $B_{j,a}$ of 
$L_j$, while $x'_{k,a} \in L'_k$ will be set to the $i_k$-th number of bucket 
$B_{k,a+y}$ of $L_k$.  \end{enumerate}

\end{enumerate}

To see the correctness of the reduction, assume there was a solution to the 
$\ksum{k}$ problem, $\{ x_{j,a_j} \}_{j\in [k]}$, and note that with probability 
$1-O(n^{-c})$, there will be an iteration for which these numbers are not mapped 
to ``overloaded" buckets. Now let $h$ be the hash function in a good iteration, 
$a = \sum_{j=1}^{k-1} h(x_{j,a_j})$, and $y$ be such that $h(x_{k,a_k})=a+y$. 
Note that by the ``almost linearity'' property, such $y \in [k]$ must exist.
Now let $(i_1,\ldots,i_k)\in [R]^k$ be such that for every $j\in [k-1]$, 
$x_{j,a_j}$ is the $i_j$-th element in bucket $B_{j,h(x_{j,a_j})}$ of $L_j$, 
while $x_{k,a_k}$ is the $i_k$-th element in bucket $B_{k,a+y}$. Now consider 
the $\convksum{k}$ instance that we get for these $(i_1,\ldots,i_k)$ and $y$, 
and consider the $\ksol{k}$ $\{ x'_{j,h(x_{j,a_j})} \}_{j \in [k-1]} \cup \{ 
x'_{k,h(x_{k,a_k})-y} \}$. Its sum will be exactly $\sum_{j\in [k]} x_{j,a_j}$, 
since $x'_{j,h(x_{j,a_j})}$ will be set to $x_{i.a_j}$, for every $j\in [k]$. 
And it will satisfy the convolution property, since $ \sum_{j=1}^{k-1} 
h(x_{j,a_j}) = a = h(x_{k,a_k}) - y$. For the other direction, any $\ksol{k}$ in 
any convolution problem is a legitimate $\ksol{k}$ in the original $\ksum{k}$ 
problem with the same sum.
Therefore, with probability $1-o(n^c)$, there is a solution iff one of the 
$\convksum{k}$ instances has a solution. 

The total running time of the reduction is $\Ot(t\cdot n^{\cftwo{k-1}} + 
(n/t)^{\cftwo{k}-\varepsilon})$. Now set $t = n^{\varepsilon}$, and note that 
when $k$ is odd, the first term is insignificant, to get a running time of 
$\Ot(n^{(1-\varepsilon)\cdot(\cftwo{k}-\varepsilon)}) = \Ot(n^{\cftwo{k} - 
\varepsilon'})$, for some $\varepsilon' >0$.


\end{proof}

\end{document}